\documentclass[11pt,twoside]{article}
\usepackage[margin=1in]{geometry}

\usepackage[utf8]{inputenc}

\usepackage{jeffe}
\usepackage[charter]{mathdesign}
\usepackage{sourcesanspro,inconsolata}
\usepackage{graphicx}
\usepackage{cite}

\usepackage{caption}
\newtheorem{theorem}{Theorem}[section]
\newtheorem{lemma}[theorem]{Lemma}
\newtheorem{corollary}[theorem]{Corollary}

\usepackage{lineno}

	\newcommand*\patchAmsMathEnvironmentForLineno[1]{%
	 \expandafter\let\csname old#1\expandafter\endcsname\csname #1\endcsname
	 \expandafter\let\csname oldend#1\expandafter\endcsname\csname end#1\endcsname
	 \renewenvironment{#1}%
	 {\linenomath\csname old#1\endcsname}%
	 {\csname oldend#1\endcsname\endlinenomath}}%
	\newcommand*\patchBothAmsMathEnvironmentsForLineno[1]{%
	 \patchAmsMathEnvironmentForLineno{#1}%
	 \patchAmsMathEnvironmentForLineno{#1*}}%
	\AtBeginDocument{%
	\patchBothAmsMathEnvironmentsForLineno{equation}%
	\patchBothAmsMathEnvironmentsForLineno{align}%
	\patchBothAmsMathEnvironmentsForLineno{flalign}%
	\patchBothAmsMathEnvironmentsForLineno{alignat}%
	\patchBothAmsMathEnvironmentsForLineno{gather}%
	\patchBothAmsMathEnvironmentsForLineno{multline}%
	}
	
\def\Torus{\mathbb{T}}

\let\eps\varepsilon
\def\Real{\mathbb{R}}
\def\Z{\mathbb{Z}}
\def\Gam{\Gamma}
\def\EMPH#1{\emph{\textbf{\boldmath #1}}}

\def\Rev{\emph{rev}}		
\def\Head{\emph{head}}		
\def\Tail{\emph{tail}}		

\title{How to Morph Graphs on the Torus\thanks{Research partially supported by NSF grants CCF-1408763, CCF-1614562, DBI-1759807, and CCF-1907612.  Portions of this work were done while the second author was visiting Utrecht University.  No animals were harmed in the making of this paper.}}
\date{}
\author{
    \href{https://cs.slu.edu/~chambers/}{Erin Wolf Chambers}
    \\SLU
    \and
	\href{http://jeffe.cs.illinois.edu}{Jeff Erickson}
	\\UIUC
	\and
	\href{https://patrickl.in/}{Patrick Lin}
	\\UIUC
	\and
	\href{http://www.sparsa.net/}{Salman Parsa}
	\\SLU
}

\begin{document}

\begin{titlepage}

\maketitle

\begin{abstract}
We present the first algorithm to morph graphs on the torus.
Given two isotopic essentially 3-connected embeddings of the same graph on the Euclidean flat torus, where the edges in both drawings are geodesics, our algorithm computes a continuous deformation from one drawing to the other, such that all edges are geodesics at all times.
Previously even the existence of such a morph was not known.
Our algorithm runs in $O(n^{1+\omega/2})$ time, where $\omega$ is the matrix multiplication exponent, and the computed morph consists of $O(n)$ parallel linear morphing steps.
Existing techniques for morphing planar straight-line graphs do not immediately generalize to graphs on the torus; in particular, Cairns' original 1944 proof and its more recent improvements rely on the fact that every planar graph contains a vertex of degree at most 5.
Our proof relies on a subtle geometric analysis of 6-regular triangulations of the torus.  We also make heavy use of a natural extension of Tutte's spring embedding theorem to torus graphs.
\end{abstract}

\thispagestyle{empty}
\setcounter{page}{0}
\end{titlepage}

\section{Introduction}

Computing a morph between two given geometric objects is a fundamental problem, with applications to questions in graphics, animation, and modeling. In general, the goal is twofold: ensure the morphs are as low complexity as possible, and ensure that the intermediate objects retain the same high level structure throughout the morph.  

Morphs between planar drawings are well studied in the topology, graph drawing, and computer graphics literature, with many variants.  A morph between two planar straight-line embeddings $\Gam_0$ and~$\Gam_1$ of the same planar graph is a continuous family of planar embeddings $\Gam_t$ parametrized by time, starting at $\Gam_0$ and ending at $\Gam_1$. In the most common formulation, all edges must be straight line segments at all times during the morph; there are then many variables of how to optimize the morph. 

In this paper, we consider the more general setting of morphs between two isotopic  embeddings of the same graph on the flat torus.  To our knowledge, ours is the first algorithm to morph graphs on any higher-genus surface.  In fact, it is the first algorithm to compute \emph{any} form of isotopy between surface graphs; existing algorithms to test whether two graphs on the same surface are isotopic are non-constructive~\cite{cm-tgis-14}.  Our algorithm outputs a morph consisting of $O(n)$ steps; within each step, all vertices move along parallel geodesics at (different) constant speeds, and all edges remain geodesics (“straight line segments”).  Our algorithm runs in $O(n^{1+\omega/2})$; the running time is dominated by repeatedly solving a linear system encoding a natural generalization of Tutte's spring embedding theorem.

\subsection{Prior Results (and Why They Don't Generalize)}
\label{ssec:prior}

Cairns \cite{c-dprc-44,c-idgc2-44} was the first to prove the existence of a straight-line continuous deformation between any two isomorphic planar straight-line triangulations.  A long series of later works, culminating in papers by Alamdari et~al. \cite{aabcd-hmpgd-17} and Kleist et~al.\cite{kklss-cimpg-19}, improved and generalized Cairns' argument to apply to arbitrary planar straight-line graphs, to produce morphs with polynomial complexity, and to derive efficient algorithms for computing those morphs.  (For a more detailed history of these results, we refer the reader to Alamdari et al. \cite{aabcd-hmpgd-17} and Roselli \cite{r-mvdg-14}.)  Cairns' inductive argument and its successors fundamentally rely on two simple  observations: (1) Every planar graph has at least one vertex of degree at most five, and (2) Every polygon with at most five vertices has at least one vertex in its visibility kernel.  Thus, every planar straight-line graph contains at least one vertex that can be collapsed to one of its neighbors while preserving the planarity of the embedding.

Unfortunately, the first of these observations fails for graphs on the torus; it is easy to construct a triangulation of the torus in which every vertex has degree~$6$.  Moreover, not every star-shaped hexagon has a vertex in its visibility kernel.  Thus, it is no longer immediate that in \emph{any} geodesic toroidal triangulation, one can move a vertex to one of its neighbors while maintaining a proper geodesic embedding.  (Indeed, the fact that we can actually collapse such an edge is the main topic of Section \ref{sec:nobadtriangulations}.)

Floater and Gotsman \cite{fg-mti-99} described an alternative method for morphing planar triangulations using a generalization of Tutte's spring-embedding theorem \cite{t-hdg-63}.  Every interior vertex in a planar triangulation can be expressed as a convex combination of its neighbors.  Floater and Gotsman's morph linearly interpolates between the coefficients of these convex combinations; Tutte's theorem implies that at all times, the interpolated coordinates are consistent with a proper straight-line embedding.  Gotsman and Surazhsky later generalized Floater and Gotsman's technique to arbitrary planar straight-line graphs \cite{sg-msfuo-01,gs-gipm-01,sg-cmcpt-01}.

At its core, Floater and Gotsman's algorithm relies on the fact that the system of linear equations expressing vertices as convex combinations of their neighbors has full rank.  An analogous system of equations describes equilibrium embeddings of graphs on the torus \cite{ggt-domam-06, el-tmcdc-20}; however, for graphs with~$n$ vertices, this linear system has $2n$ equations over $2n$ variables (the vertex coordinates), but its rank is only $2n-2$.  If the linear system happens to have a solution, that solution is consistent with a proper embedding \cite{c-crgtd-91, d-eppgc-04, l-dafe-04, ggt-domam-06}; unfortunately, the system is not always solvable.

When the coefficients associated with each edge are symmetric, the linear system has a two-dimen\-sional set of solutions, which correspond to proper embeddings that differ only by translation.  (See our Theorem \ref{Th:tutte-torus} below.)  Thus, if two given triangulations can both be described by symmetric coefficients, linearly interpolating those coordinates yields an isotopy~\cite{cpv-tbmai-03}.  Otherwise, however, even if the initial and final coefficient vectors are feasible, weighted averages of those coefficients might not be.  Steiner and Fisher \cite{sf-ppc2m-04} modify the linear system by fixing one vertex, restoring full rank.  However, while the solution to this linear system always describes a geodesic drawing of the graph, edges in that drawing can cross.  In either setting, linearly interpolating the edge coefficients does not yield a morph.

Both of these approaches produce planar morphs that require high numerical precision to describe exactly.  Barerra-Cruz et al. \cite{bhl-msdpt-19} describe an algorithm to morph between two isomorphic \emph{weighted Schnyder drawings} of the same triangulation, each determined by a Schnyder wood together with an assignment of positive weights to the faces.  The resulting morph consists of $O(n^2)$ steps, where after each step, all vertices lie on a $6n\times 6n$ integer grid.  The algorithm relies crucially on the fact that the set of Schnyder woods of a planar triangulation is a distributive lattice \cite{f-lsfpg-04}.  Despite some initial progress by Barerra-Cruz \cite{b-mpt-14}, it is still an open question whether this algorithm can be extended to arbitrary planar triangulations, or even to arbitrary planar straight-line graphs.  Beyond that, it is also not clear whether this result can be extended to toroidal graphs.  Aleardi et al. \cite{afl-swhgt-09} and Gonçalves and Lévêque \cite{gl-tmswo-14} describe natural generalizations of Schnyder woods to graphs on the torus; however, the Schnyder woods (or 3-orientations) of a toroidal triangulation do not form a distributive lattice. 

Considerably less is known about morphing graphs on higher-genus surfaces.  Like earlier planar morphing algorithms, our algorithm follows the same inductive strategy as several earlier algorithms for transforming combinatorial embeddings into geodesic embeddings on the torus \cite{mgk-rnm-79,m-srmto-96,kns-dgt-01}.  Our algorithm most closely resembles an algorithm of Kocay et al. \cite{kns-dgt-01}, which transforms any essentially 3-connected toroidal embedding into an isotopic  geodesic embedding, by repeatedly collapsing vertices with degree at most 5 until the embedding becomes a 6-regular triangulation.

\subsection{Outline of Our Results}

We begin by reviewing relevant definitions and background in Section~\ref{sec:background}.
Most importantly, we review a natural generalization of Tutte’s spring embedding theorem \cite{t-hdg-63} to graphs on the flat torus, first proved by Y.~Colin de Verdière~\cite{c-crgtd-91}; see Theorem \ref{Th:tutte-torus}.
We present a technical overview of our contributions in Section~\ref{sec:pseudomorph}, deferring details to later sections for clarity.

Like many previous planar morphing papers, most of our paper is devoted to computing \emph{pseudomorphs} between \emph{triangulations}.  A pseudomorph is a continuous deformation in which vertices are allowed to coincide during the motion but edges are not allowed to cross.  Our pseudomorph algorithm uses two different operations that reduce the complexity of the graph: \emph{direct collapses}, which move one vertex to one of its neighbors, and \emph{spring collapses}, which increase the weight of one edge to infinity while maintaining an equilibrium embedding, as described by Theorem \ref{Th:tutte-torus}.  The heart of our result is a novel analysis of 6-regular toroidal triangulations in Section~\ref{sec:nobadtriangulations}, which implies that every non-trivial toroidal triangulation contains at least one edge that can be directly collapsed without introducing any crossings.  We regard this analysis as the main technical contribution of our paper.  We describe and analyze spring collapses in Section~\ref{sec:eqpseudo}, again relying on Theorem \ref{Th:tutte-torus}.  We describe and analyze the base case of our pseudomorph algorithm in Section~\ref{sec:zippers}: a special class of triangulations we call \emph{zippers}, where every vertex is incident to a loop.

In Section~\ref{sec:depseudo}, we show that a mild generalization of techniques from Alamdari et al.~\cite{aabcd-hmpgd-17} can be used to perturb our pseudomorph into a proper morph; this perturbation technique gives us our final morphing algorithm for triangulations.  Finally, in Section~\ref{sec:allgraphs}, we describe a simple reduction from morphing essentially 3-connected geodesic toroidal embeddings to morphing triangulations, again using Theorem \ref{Th:tutte-torus}.  We conclude in Section~\ref{sec:conclusions} with some open problems and future directions to consider.

\section{Background and Definitions}
\label{sec:background}

\subsection{The Flat Torus}

The \EMPH{flat torus} $\Torus$ is the metric space obtained by identifying opposite sides of the unit square $[0,1]^2$ in the Euclidean plane via $(x,0) \sim (x,1)$ and $(0,y) \sim (1,y)$.  See Figure~\ref{fig:two-vertices-torus}.  Equivalently, $\Torus$ is the quotient space $\Torus = \R^2/\Z^2$, obtained by identifying every pair of points whose $x$- and $y$-coordinates differ by integers.  The function $\pi \colon \R^2 \to \Torus$ defined by $\pi(x,y) = (x\bmod1, y\bmod 1)$ is called the \emph{covering} map or \EMPH{projection} map.

A \EMPH{geodesic} in $\Torus$ is the projection of any line segment in $\R^2$; geodesics are the natural analogues of “straight line segments” on the flat torus.  We emphasize that a geodesic is \emph{not} necessarily the shortest path between its endpoints; indeed, there are infinitely many geodesics between any two points  on $\Torus$.  A \EMPH{closed geodesic} in $\Torus$ is any geodesic whose endpoints coincide; the two ends of any closed geodesic are locally collinear.

\subsection{Toroidal Embeddings} 

Geodesic toroidal drawings are the natural generalizations of straight-line planar graphs to the flat torus.  Formally, a \EMPH{geodesic toroidal drawing} $\Gam$ of a graph $G$ is a mapping of vertices to distinct points of $\Torus$ and edges to non-intersecting geodesics between their endpoints.  Following standard usage in topology, we refer to any such drawing as \EMPH{embedding}, to emphasize that edges do not cross.\footnote{Formally, an embedding is a continuous injective map from the graph (as a topological space) to the torus $\Torus$.  We note that this usage differs from standard terminology in many other graph drawing papers, where “embedding” refers to either a homeomorphism class of (not necessarily injective) drawings or a rotation system.}

A \EMPH{homotopy} between two (not necessarily injective) drawings $\Gamma_0$ and $\Gamma_1$ of the same graph $G$ is a continuous function $H\colon [0,1]\times G\to \Torus$ where $H(0,) = \Gamma_0$ and $H(1,)= \Gamma_1$.  A cycle on $\Torus$ is \EMPH{contractible} if it is homotopic to a single point and \EMPH{non-contractible} otherwise.  A homotopy is an \EMPH{isotopy} if each intermediate function $H(t,)$ is injective.  In other words, an isotopy is a continuous family of embeddings $(\Gam_t)_{t \in [0,1]}$ that interpolates between $\Gam_0$ and~$\Gam_1$.  (Edges in these intermediate embeddings $\Gamma_t$ are not necessarily geodesics.)

\begin{figure}[ht]
\centering
\includegraphics[scale=0.6]{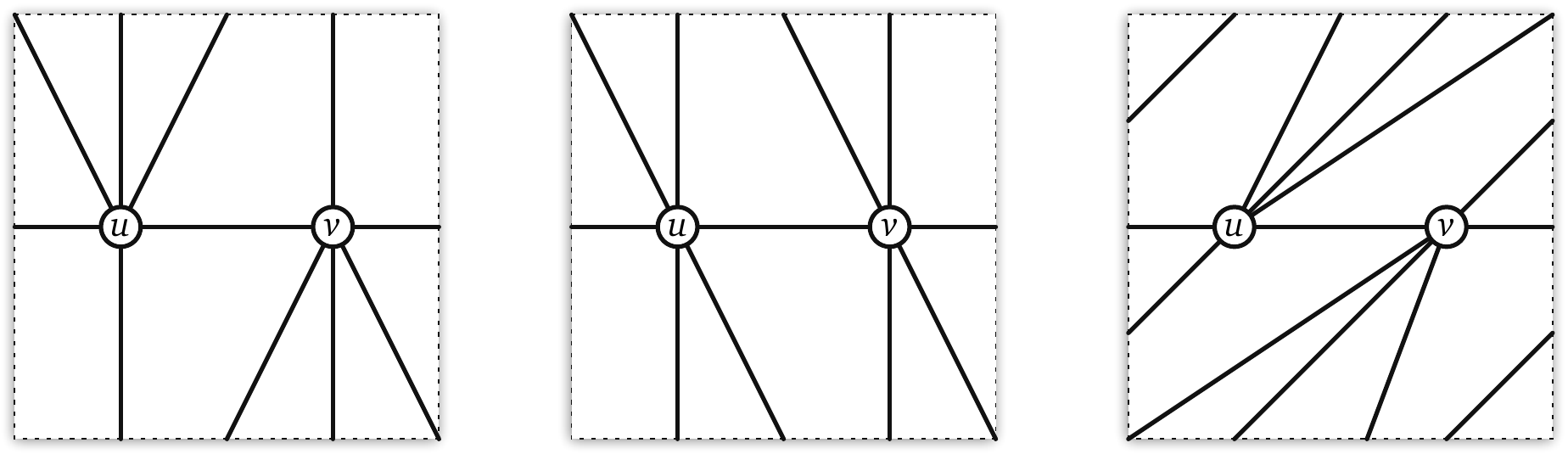}
\caption{Three combinatorially equivalent but non-isotopic geodesic toroidal triangulations with parallel edges and loops.  Opposite edges of the square are identified.}
\label{fig:two-vertices-torus}
\end{figure}

Two toroidal embeddings of the same graph need not be isotopic, even if they have the same rotation system; see Figure~\ref{fig:two-vertices-torus}. 
A recent algorithm of É.~Colin de Verdière and de Mesmay~\cite{cm-tgis-14} can decide whether two toroidal drawings of the same graph are isotopic in linear time; we describe an arguably simpler linear-time algorithm in Appendix \ref{A:coordinates}.  However, neither of these algorithms actually construct an isotopy if one exists; rather, they check whether the two embeddings satisfy certain topological properties that characterize isotopy \cite{l-ctp1c-74a,l-ctp1c-74b,l-cdp1c-84}.

We explicitly consider embeddings of graphs with parallel edges and loops.  In every geodesic toroidal embedding, every loop is non-contractible (since otherwise it would be a single point), and no two parallel edges are homotopic (since otherwise they would coincide).  In this paper, we consider \emph{only} geodesic embeddings; we occasionally omit the word “geodesic” when it is clear from context.

The \EMPH{universal cover} $\widetilde{\Gam}$ of a geodesic toroidal embedding $\Gam$ is the unique infinite straight-line plane graph whose projection to $\Torus$ is $\Gam$; that is, the projection of any vertex, edge, or face of $\widetilde{\Gam}$ is a vertex, edge, or face of $\Gam$, respectively.
A \EMPH{lift} of any vertex $u$ in $\Gam$ is any vertex in the preimage $\pi^{-1}(u) \subset V(\widetilde{\Gam})$. 
Similarly, each edge of $\Gam$ lifts to an infinite lattice of parallel line segments in~$\R^2$, and each face lifts to an infinite lattice of congruent polygons.

\begin{figure}[ht]
\centering
\includegraphics[scale=0.3]{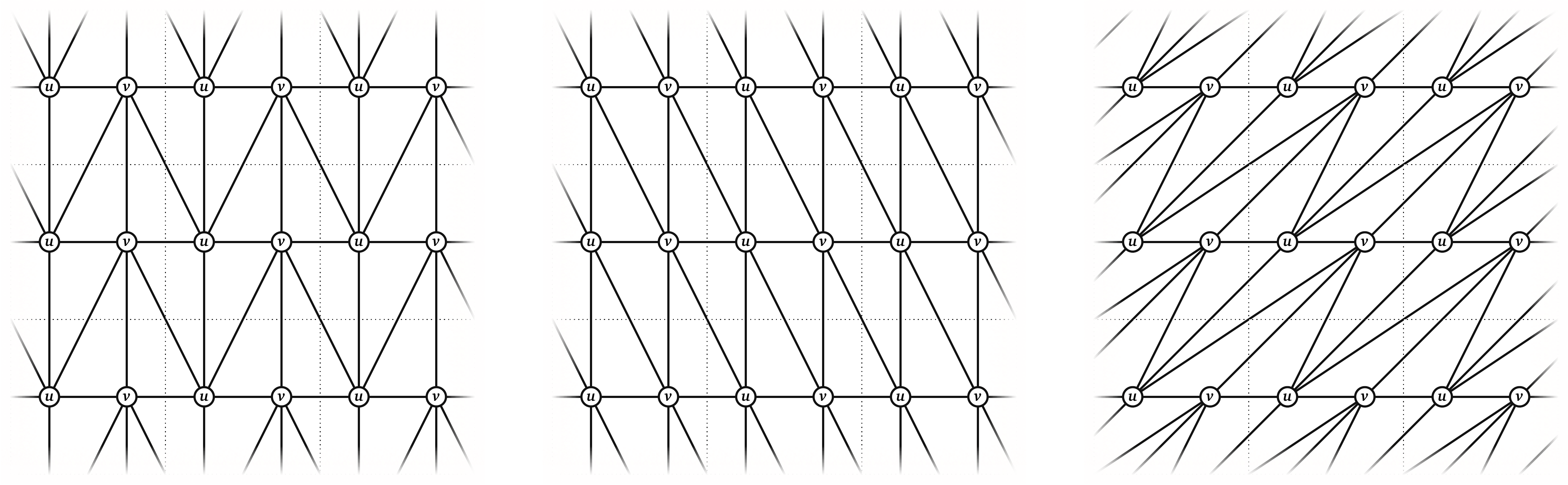}
\caption{Universal covers of the geodesic embeddings from Figure \ref{fig:two-vertices-torus}.}
\label{F:universal}
\end{figure}

The \EMPH{link} of a vertex $\widetilde{u}$ in the universal cover $\widetilde{\Gam}$ is the simple polygon formed by the boundary of the union of the (closed) faces incident to $\widetilde{u}$; the vertices of the link are the neighbors of $\widetilde{u}$.
We emphasize that when projecting a link down to the flat torus, the vertices and edges of the link need not remain distinct; see  Figure~\ref{fig:two-vertices-torus} for an example.  For a vertex $u$ in $\Gam$, we informally write ``link of $u$'' to refer to the link of an arbitrary lift $\widetilde{u}$ of $u$, and similarly for edges of $\Gam$. Because the links of any two lifts are congruent, any property proven about one lift applies to all of the others.

Geometric properties of geodesics, polygons, and embeddings on the flat torus are defined by projection from the universal cover.  For example, the angle between two edges (or geodesics)~$e$ and~$e'$ at a common vertex $u$ is equal to the angle between lifts $\widetilde{e}$ and~$\widetilde{e}'$ at a common lift $\widetilde{u}$.  Similarly, the cyclic order of edges around a vertex $u$ of $\Gam$ is the cyclic order of the corresponding edges around an arbitrary lift $\widetilde{u}$.   In particular, if~$u$ is incident to a loop, that loop appears twice in cyclic order around~$u$, and each lift $\widetilde{u}$ of $u$ is incident to two different lifts of that loop.  Finally, convex or reflex angles in the link of a vertex in $\Gamma$ are projections of convex or reflex angle in the link of an arbitrary lift~$\widetilde{u}$.

A toroidal embedding $\Gam$ is a \EMPH{triangulation} if every face of $\Gam$ is bounded by three (not necessarily distinct) edges, or equivalently, if its universal cover $\widetilde{\Gam}$ is a planar triangulation.  In particular, we do not insist that triangulations are simplicial complexes.  Every \emph{geodesic} toroidal embedding $\Gam$ is \EMPH{essentially simple}, meaning its universal cover $\widetilde{\Gam}$ is a planar embedding of a simple (albeit infinite) graph.  A~geodesic toroidal drawing $\Gam$ is \EMPH{essentially 3-connected} if its universal cover~$\widetilde{\Gam}$ is 3-connected~\cite{m-cpmec-97,m-cpmpt-97,mr-tvrmt-98,ms-bns-03,gl-tmswo-14}; every geodesic triangulation is essentially 3-connected.




\subsection{Coordinates and Crossing Vectors}
\label{SS:canon}

To represent an arbitrary straight-line embedding of a graph in the plane, it suffices to record the coordinates of each vertex; each edge in the embedding is the unique line segment between its endpoints.  However, vertex coordinates alone are not sufficient to specify a toroidal embedding; intuitively, we must also specify how the edges of the graph wrap around the surface.

Formally, we regard each edge of the graph $G$ as a pair of opposing \emph{half-edges} or \EMPH{darts}, each directed from one endpoint, called the \EMPH{tail}, toward the other endpoint, called the \EMPH{head}.  We write $\Rev(d)$ to denote the reversal of any dart $d$; thus, for example, $\Head(\Rev(d)) = \Tail(d)$ and $\Rev(\Rev(d)) = d$ for every dart~$d$.

We can represent any geodesic embedding of any graph $G$ onto the torus by associating a \EMPH{coordinate} vector $p(v) \in [0,1)^2$ with every vertex $v$ of $G$ and a \EMPH{crossing} vector $x(d) \in \Z^2$ with every dart $d$ of~$G$.  The coordinates of a vertex specify its position in the unit square; to remove any ambiguity, we assign points on the boundary of the unit square coordinates on the bottom and/or left edges.  The crossing vector of a dart records how that dart crosses the boundaries of the unit square.  Specifically, the first coordinate of $x(d)$ is the number of times $d$ crosses the vertical boundary to the right (with negative numbers counting leftward crossings), and the second coordinate of $x(d)$ is the number of times $d$ crosses the horizontal boundary upward (with negative numbers counting downward crossings).  Crossing vectors are anti-symmetric: $x(\Rev(d)) = -x(d)$ for every dart~$d$.  See Figure \ref{F:crossing}.

\begin{figure}[ht]
\centering
\includegraphics[scale=0.6]{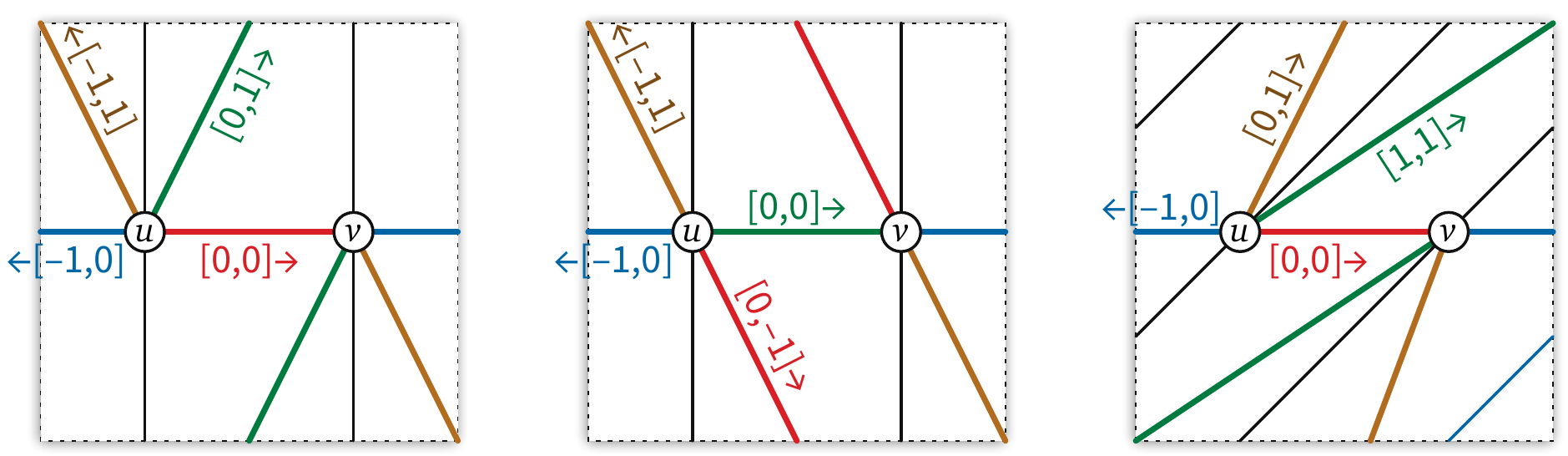}
\caption{The geodesic embeddings from Figure \ref{fig:two-vertices-torus}, showing the crossing vectors of all four darts from $u$ to $v$.}
\label{F:crossing}
\end{figure}


Crossing vectors and their generalizations have been used in several previous algorithms for surface graphs \cite{ew-gohhg-05,en-mcsnc-11,ccelw-scsih-08,cen-mcshc-09,cfn-csmcg-14,efl-hmcpf-18} and simplicial complexes \cite{bccdw-aswhb-12,dsw-alshb-10,dlw-eacmh-18} to encode the homology classes of cycles.  Crossing vectors are also equivalent to the \emph{translation vectors} traditionally used to model periodic (or “dynamic”) graphs \cite{kmw-ocure-67, c-spg-78, ks-dcdgp-88, cm-spadc-93, o-spdg-84, rk-riati-88, w-pdmia-67, is-tcigw-87, is-scpzs-90, i-tdgtv-87} and more recently used to model periodic bar-and-joint frameworks \cite{bs-lsppf-15,bs-pff-10}.

In principle, our morphing algorithm can be modified to update the coordinates of any vertex $v$ and the crossing vectors of darts incident to $v$ whenever $v$ crosses the boundary of the unit square, with only a small penalty in the running time.  But in fact, this maintenance is not necessary; it suffices to modify \emph{only} the vertex coordinates, keeping all crossing vectors fixed throughout the entire morph, even when vertices cross the boundary of the unit square.  We describe how to interpret toroidal embeddings with these more relaxed coordinates in Appendix \ref{A:coordinates}.

\subsection{Equilibrium Embeddings}

We make frequent use of the following natural generalization of Tutte's “spring embedding” theorem for 3-connected planar graphs \cite{t-hdg-63}:

\begin{theorem}[Y.~Colin de Verdière~\cite{c-crgtd-91}; see also \cite{d-eppgc-04, l-dafe-04, ggt-domam-06}]
Let $\Gam$ be any essentially 3-connected geodesic toroidal drawing, where each edge $e$ has an associated weight $\lambda(e) > 0$.  Then $\Gam$ is isotopic to a geodesic embedding $\Gam_*$ in $\Torus$ such that every face is convex and each vertex is the weighted center of mass of its neighbors; moreover, this equilibrium embedding is unique up to translation.
\label{Th:tutte-torus}
\end{theorem}

The equilibrium embedding $\Gam_*$ can be computed by solving the following linear system for the vertex coordinates $p_*(v)$, treating the crossing vectors $x(d)$ and weights $\lambda(d)$ as constants~\cite{c-crgtd-91,sf-ppc2m-04,ggt-domam-06,el-tmcdc-20}:
\begin{equation}
	\sum_{\Tail(d) = v} \lambda(d) \cdot
		\left(\strut p_*(\Head(d)) + x(d) - p_*(v)\right) = (0,0)
	\qquad \text{for every vertex $v$}
	\tag{$\star$}\label{Eq:Tutte}
\end{equation}
Here $\lambda(d) = \lambda(\Rev(d))$ is the weight of the edge containing dart $d$.   
This linear system has a two-dimensional set of solutions, which differ by translation \cite{c-crgtd-91,ggt-domam-06}; we can remove this ambiguity by fixing $p_*(r) = (0,0)$ for some arbitrary \emph{root} vertex $r$ \cite{sf-ppc2m-04}.  Some vertex coordinates $p_*(v)$ in the solution to this system may lie outside the unit square; as we explain in Appendix \ref{A:coordinates}, it is possible to move all coordinates back into the unit square by appropriately adjusting the crossing vectors, but in fact no such adjustment is necessary.

The support of this linear system is the toroidal graph $G$, which has balanced separators of size $O(\sqrt{n})$~\cite{ad-lapeg-96}.  Thus, the linear system can be solved in $O(n^{\omega/2})$ time using the generalized nested dissection technique of Lipton et al.~\cite{lrt-gnd-79}, where $\omega < 2.37287$ is the matrix multiplication exponent~\cite{l-ptfmm-14}. 

\subsection{Morphs and Pseudomorphs}

A \EMPH{morph} between two isotopic geodesic toroidal drawings $\Gam_0$ and $\Gam_1$ is a continuous family of geodesic drawings $(\Gam_t)_{t \in [0,1]}$ from $\Gam_0$ to $\Gam_1$; in other words, a morph is a geodesic isotopy between $\Gam_0$ and $\Gam_1$.   Any morph is completely determined by the continuous motions on the vertices; geodesic edges update in the obvious way.  

A morph is \EMPH{linear} if every vertex moves along a geodesic from its initial position to its final position at a uniform speed, and a morph is \EMPH{parallel} if all vertices move along parallel geodesics, that is, along projections of parallel line segments.\footnote{The paper of Alamdari et al.~\cite{aabcd-hmpgd-17} and its predecessors~\cite{bhl-mpgdu-14,addfp-mpgdo-14} call these “unidirectional" morphs; however, this term suggests incorrectly that vertices cannot move in \emph{opposite} directions.}
In this paper, we construct morphs that consist of a sequence of $O(n)$ parallel linear morphs.  Every morph of this type can be specified by a sequence of isotopic geodesic toroidal embeddings $\Gam_0,\ldots,\Gam_k$ and for each index $i$, a set of parallel geodesics connecting the vertices of $\Gam_i$ to corresponding vertices of $\Gam_{i+1}$.

Like many previous planar morphing algorithms, our morphing algorithm first constructs a \EMPH{pseudomorph}, which is a continuous family of drawings in which edges remain geodesics, vertices may become coincident, but edges never cross.  The most basic ingredient in our pseudomorph is an \EMPH{edge collapse}, which moves the endpoints of one edge together until they coincide.\footnote{This procedure is often called edge \emph{contraction} in planar graph morphing literature; we use the term “collapse” to avoid any confusion with the topological notion of contractible cycles.} Collapsing an edge also collapses the faces on either side of that edge to geodesics.  Our final morph is obtained by carefully perturbing a pseudomorph consisting of edge collapses and their reversals.

\section{Technical Overview}
\label{sec:pseudomorph}

In this section, we give a technical overview of our results: at a high level, first we develop an algorithm to compute a \emph{pseudomorph} between geodesic toroidal \emph{triangulations}, and then we describe how to use this algorithm to compute a morph between essentially $3$-connected geodesic toroidal embeddings.  Here we give only a brief overview of several necessary tools; each of these components is developed in detail in a later section of the paper.

\subsection{List of Ingredients}

We begin with some essential subroutines and structural results.

\subsubsection{Direct collapses}

Following Cairns' approach~\cite{c-dprc-44} and its later derivatives~\cite{aabcd-hmpgd-17,aacbf-mpgdw-13,bhl-mpgdu-14,afpr-mpgde-13,addfp-mpgdo-14}, a \EMPH{direct collapse} consists of moving a vertex $u$ along some edge to another vertex $v$, at uniform speed, until $u$ and $v$ coincide, keeping all other vertices fixed.  
To simplify our presentation, we require that the moving vertex $u$ is not incident to a loop.  We informally call a vertex \EMPH{good} if it is not incident to a loop and it can be directly collapsed to one of its neighbors without introducing any edge crossings, and \EMPH{bad} otherwise; see Section~\ref{sec:nobadtriangulations} for a simple geometric characterization of good vertices.

As noted by Cairns~\cite{c-dprc-44}, every vertex of degree at most $5$ not incident to a loop is good; indeed, this fact, along with the fact that every planar graph has a vertex of degree at most $5$, forms the basis of Cairns' approach and its derivatives for computing (pseudo)morphs between straight-line planar drawings. We prove in Lemma~\ref{L:loop6} that in any geodesic toroidal triangulation, a vertex of degree at most $5$ cannot be incident to a loop, so in fact all vertices of degree at most $5$ are good.

On the other hand, Euler's formula implies that the average degree of a toroidal graph is exactly~$6$, and there are simple examples of degree-$6$ vertices that are bad.  Morphing between torus graphs thus requires new techniques to handle the special case of $6$-regular triangulations.

\subsubsection{6-regular triangulations}

We then prove in Lemma~\ref{L:loopmeanszipper} that if a 6-regular toroidal triangulation contains a loop, then in fact \emph{every} vertex is incident to a loop; we call this special type of triangulation a \EMPH{zipper}.  Because all the loops in any non-trivial zipper are parallel, morphing between isotopic zippers $Z_0$ and $Z_1$ turns out be straightforward.  If the zippers have only one vertex, they differ by a single translation.  Otherwise, two parallel linear morphs suffice, first sliding the loops of~$Z_0$ to coincide with the corresponding loops in~$Z_1$, and then rotating the loops to move vertices and non-loop edges to their locations in $Z_1$.  We describe and analyze zippers in detail in Section~\ref{sec:zippers}.

We analyze $6$-regular triangulations without loops in detail in Section~\ref{sec:nobadtriangulations}; we regard this analysis as the main technical contribution of the paper.  Averaging arguments imply that in any 6-regular triangulation where every vertex is bad, in short a \emph{bad triangulation}, every vertex link is one of two specific non-convex hexagons that we call \emph{cats} and \emph{dogs}, illustrated in Figure \ref{F:catdog-anatomy}.  Analysis of how cats and dogs can overlap implies that any bad triangulation must contain a non-contractible cycle of vertices (each of whose link is a cat) that consistently turns in the same direction at every vertex, and therefore has non-zero turning angle, contradicting the fact that the total turning angle of every non-contractible cycle on the flat torus is zero \cite{r-let-59}.  We conclude that bad triangulations do not exist; \emph{every} geodesic toroidal triangulation that is not a zipper contains at least one good vertex.

\subsubsection{Equilibria and Spring Collapses}

Suppose we are given two isotopic geodesic triangulations $\Gam_0$ and $\Gam_1$ that are not zippers.  Our analysis implies that $\Gam_0$ and $\Gam_1$ each contain at least one good vertex.  However, it is possible that no vertex is good in both $\Gamma_0$ and $\Gamma_1$.  More subtly, even if some vertex $u$ is good in both triangulations, that vertex may be collapsible along a unique edge $e_0$ in $\Gamma_0$ but along a different unique edge $e_1$ in $\Gamma_1$.

The second problem also occurs for straight-line planar embeddings.  Cairns' solution to this problem was to introduce an intermediate triangulation $\Gam_{1/2}$ in which $u$ can be collapsed along both $e_0$ and~$e_1$.  Recursively constructing pseudomorphs from $\Gam_0$ to $\Gam_{1/2}$ and from $\Gam_{1/2}$ to $\Gam_1$ yields a pseudomorph from~$\Gam_0$ to $\Gam_1$ with exponentially many steps.  Subsequent refinements of Cairns' approach, culminating in the work of Alamdari et al.~\cite{aabcd-hmpgd-17} and later improvement by Kleist et al.~\cite{kklss-cimpg-19}, obtained a pseudomorph with only polynomial complexity by finding clever ways to avoid this intermediate triangulation.

Our algorithm does introduce one intermediate triangulation, but still avoids the exponential blowup of Cairns' algorithm.  Specifically, we use an equilibrium triangulation $\Gam_*$ isotopic to $\Gam_0$ and $\Gam_1$, as given by Theorem~\ref{Th:tutte-torus}.  A vertex that is good in $\Gam_0$ might still be bad in $\Gam_*$, so instead of applying a direct collapse to $\Gam_*$, we introduce a novel method for collapsing edges in an equilibrium embedding in Section~\ref{sec:eqpseudo}.  Intuitively, we continuously increase the weight of an arbitrary edge $e$ to infinity, while maintaining the equilibrium triangulation given by Theorem~\ref{Th:tutte-torus}.  This \EMPH{spring collapse} moves the endpoints of $e$ together, just like a direct collapse.  By analyzing the solutions to the equilibrium linear system~\eqref{Eq:Tutte}, we show that a spring collapse is a parallel pseudomorph, and in fact can be simulated by an equivalent parallel \emph{linear} pseudomorph.  Moreover, this parallel linear pseudomorph can be computed by solving a single instance of system~\eqref{Eq:Tutte}.
%

\subsection{Recursive Pseudomorph Between Triangulations}
\label{ssec:triangulationpseudomorph}

We are now ready to describe our recursive algorithm to compute a pseudomorph between two isotopic geodesic triangulations $\Gam_0$ and $\Gam_1$.  We actually explain how to  compute a pseudomorph $\Psi_0$ from $\Gam_0$ to an isotopic equilibrium triangulation $\Gam_*$.  The same algorithm gives a pseudomorph $\Psi_1$ from $\Gam_1$ to $\Gam_*$, and concatenating $\Psi_0$ with the reversal of $\Psi_1$ yields the desired pseudomorph from $\Gam_0$ to $\Gam_1$.

If $\Gam_0$ is a zipper, we morph directly between $\Gam_0$ and the equilibrium zipper $\Gam_*$ using at most two parallel linear morphs.  This is the base case of our recursive algorithm.

If $\Gam_0$ is not a zipper, then it contains a good vertex $u$.  By definition, $u$ can be directly collapsed along some edge $e$ to another vertex $v$, without introducing edge crossings.  This direct collapse gives us a parallel linear pseudomorph from $\Gam_0$ to $\Gam_0'$, a geodesic toroidal triangulation whose underlying graph $G'$ has $n-1$ vertices.  On the other hand, performing a spring collapse in $\Gam_*$ by increasing the weight of the same edge~$e$ to $\infty$ leads to a drawing where $u$ and $v$ coincide, that is, an equilibrium triangulation $\Gam_*'$ of~$G'$ that is isotopic to $\Gam_0'$.  Finally, because $\Gam_0'$ and $\Gam_*'$ are isotopic embeddings of the same graph $G'$, we can compute a pseudomorph from $\Gam_0'$ to $\Gam_*'$ recursively.

Our full pseudomorph from $\Gam_0$ to $\Gam_*$ thus consists of the direct collapse from $\Gam_0$ to~$\Gam_0'$, followed by the recursive pseudomorph from $\Gam_0'$ to $\Gam_*'$, followed by the reverse of the spring collapse from $\Gam_*$ to $\Gam'_*$.  See Figure~\ref{F:summary}.

\begin{figure}[ht]
\centering
\includegraphics[width=\textwidth]{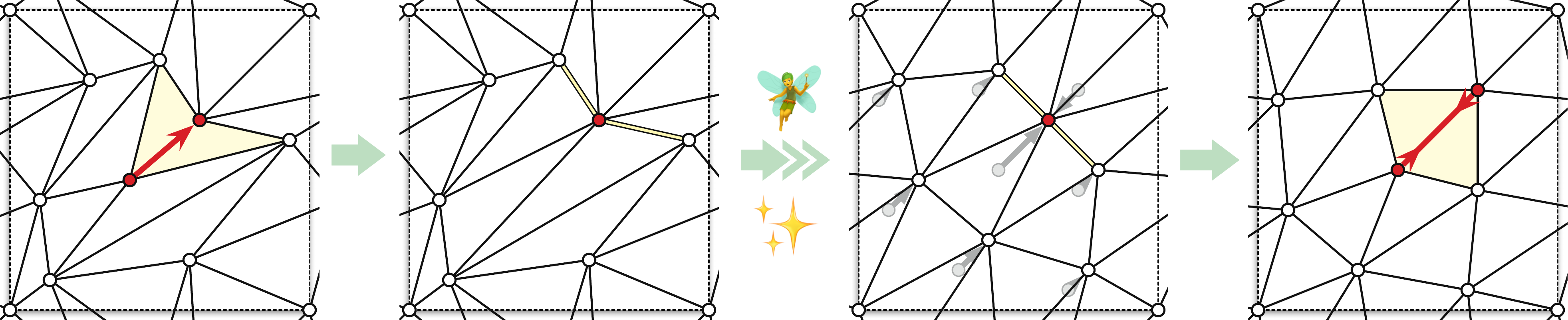}
\caption{Our pseudomorph consists of a direct collapse, a recursive pseudomorph, and a reversed spring collapse.}
\label{F:summary}
\end{figure}

Altogether our pseudomorph consists of $O(n)$ parallel linear pseudomorphs: at most $n$ direct collapses to reach a zipper, $O(1)$ parallel linear pseudomorphs to reach an equilibrium zipper, and finally at most $n$ reversed spring collapses.  The time to compute the overall pseudomorph is dominated by the time needed to compute the $O(n)$ equilibrium embeddings; the overall running time is $O(n^{1+\omega/2})$.

\subsection{Morphing Between Embeddings}

In Section~\ref{sec:depseudo}, we show that a technique introduced by Alamdari et al.~\cite{aabcd-hmpgd-17} for perturbing planar pseudomorphs into morphs can be applied with only minor modifications to our toroidal pseudomorphs, to produce morphs between geodesic toroidal triangulations.  This perturbation step adds only $O(n^2)$ overhead to the running time of our algorithm, and the resulting morphs  consist of $O(n)$ parallel linear morphs.

Finally, in Section~\ref{sec:allgraphs}, we describe a simple reduction from morphing arbitrary essentially $3$-connected embeddings to the special case of morphing triangulations.  In brief, given any triangulation of any essentially $3$-connected geodesic toroidal embedding $\Gam$, the isotopic equilibrium embedding $\Gam_*$  given by Theorem~\ref{Th:tutte-torus} can be triangulated in the same way, and any morph between the two triangulations induces a morph between $\Gam$ and $\Gam_*$. This reduction preserves the number of parallel linear morphs and adds only $O(n^{\omega/2})$ overhead to the running time.

In summary, given any two isotopic essentially $3$-connected geodesic toroidal embeddings, we can construct a morph between them, consisting of $O(n)$ parallel linear morphs, in $O(n^{1+\omega/2})$ time.

\section{Cats and Dogs}
\label{sec:nobadtriangulations}

In this section, we present the core of our technical analysis: the proof that every geodesic toroidal triangulation without loops has a (directly) collapsible edge.

The \EMPH{visibility kernel} of a simple polygon $P$ is the set of all points in $P$ that can “see” all of $P$; more formally, the visibility kernel is $\Setbar{p\in P}{{pq\subseteq P} \text{~for all~} {q\in P}}$.  The visibility kernel is always convex.
If~$P$ is the link of a vertex $v$ in a geodesic triangulation, then $v$ must lie in the visibility kernel of $P$.
We call a simple polygon \EMPH{good} if its visibility kernel contains a vertex of the polygon, and \EMPH{bad} otherwise.  All triangles, quadrilaterals, and pentagons are good \cite{c-dprc-44}, but some hexagons are bad;  Figure~\ref{fig:catdogkernel} shows several examples.

\begin{figure}[ht]
\centering
\includegraphics[scale=0.5]{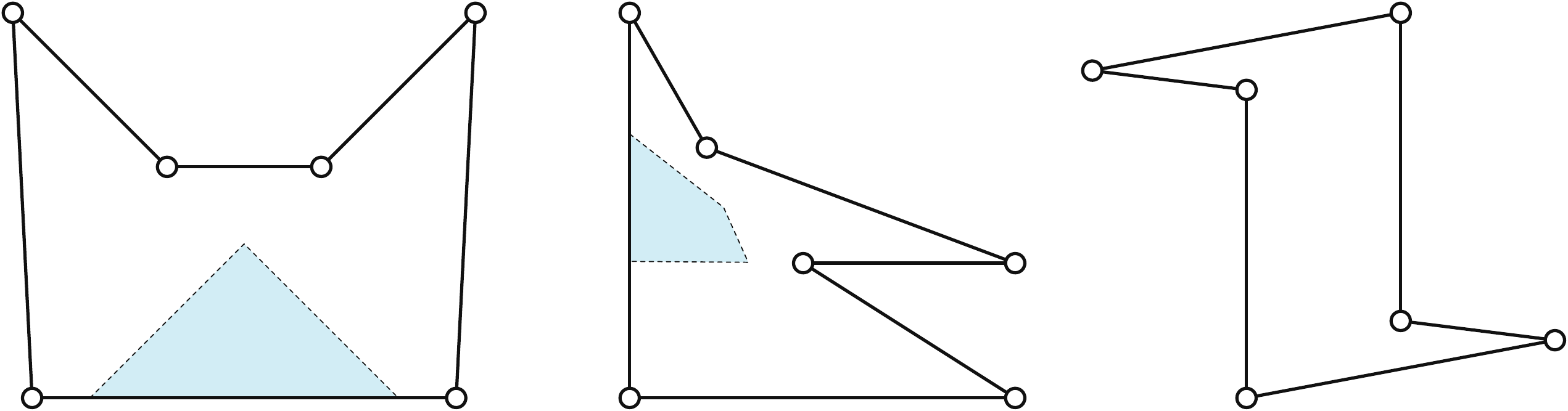}
\caption{Three bad hexagons.  Visibility kernels are shaded in blue; the third visibility kernel is empty.}
\label{fig:catdogkernel}
\end{figure}

A vertex $u$ of a geodesic toroidal triangulation is \EMPH{good} if it is not incident to a loop, and the link of any (and thus every) lift $\tilde{u}$ of $u$ in~$\widetilde\Gam$ is good, and \EMPH{bad} otherwise.  A good vertex can be safely collapsed to any neighbor in the visibility kernel of its link.  Finally, a geodesic toroidal triangulation~$\Gam$ with no loops is \EMPH{good} if it contains at least one good vertex, and \EMPH{bad} otherwise.  The main result of this section is that bad triangulations do not exist; that is, every geodesic triangulation with no loops is good.




\begin{lemma}\label{l:6regular2reflex}
Any bad triangulation is $6$-regular, and the link of each vertex has exactly two reflex vertices.
\end{lemma}

\begin{proof}
Every vertex in a bad triangulation must have degree at least $6$, because every vertex with degree at most $5$ is good~\cite{c-dprc-44}.  On the other hand, Euler's formula for the torus implies that the average degree is exactly $6$.  It follows that every bad triangulation is $6$-regular.

If a simple polygon is convex or has exactly one reflex vertex, then it is good: the visibility kernel of a convex polygon is the polygon itself, and if there is exactly one reflex vertex, then the reflex vertex is in the visibility kernel.  So each link in a bad triangulation has \emph{at least} two reflex vertices.  We argue next that the average number of reflex vertices per link is \emph{at most} two, which implies that every link has exactly two reflex vertices.

%

Let $\Gamma$ be any (not necessarily bad) 6-regular geodesic triangulation.  A \emph{corner} of a vertex $v$ in $\Gamma$ is the angle between two edges that are adjacent in cyclic order around $v$.  If $v$ is a vertex of the link of another vertex $x$, then the link of $x$ contains exactly two adjacent corners of $v$.  Moreover, if $v$ is a reflex vertex of the link of $x$, those two corners sum to more than half a circle.  Thus, if $v$ is reflex in the link of two neighbors $x$ and $y$, the links of $x$ and $y$ must share a corner of $v$, which implies that edges $vx$ and $vy$ are adjacent in cyclic order around $v$; see \ref{F:2reflex}.  If $v$ were reflex in the link of three neighbors $x$,~$y$, and $z$, then all three edge pairs $vx$, $vy$ and $vx$, $vz$ and $vy$, $vz$ would be adjacent around $v$, which is impossible because $v$ has degree $6$.

\begin{figure}[ht]
    \centering
    \includegraphics[scale=0.5]{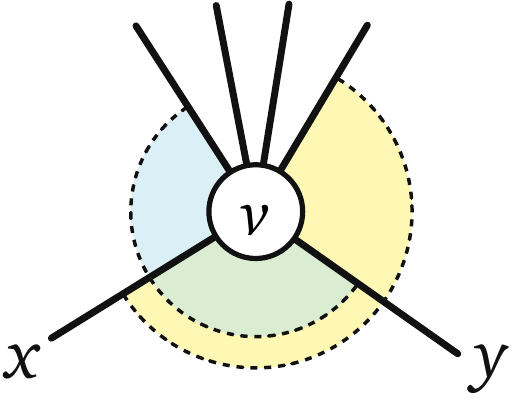}
    \caption{Each vertex in a 6-regular triangulation is reflex in at most two links}
    \label{F:2reflex}
\end{figure}

Thus, each vertex in $\Gamma$ is a reflex vertex of the links of at most two other vertices.  It follows that the average number of reflex vertices in a link is at most two, which completes the proof.
\end{proof}

\begin{lemma}\label{lem:nomice}
A bad hexagon with two reflex vertices separated by two convex vertices is never the link of any vertex.
\end{lemma}

\begin{proof}
Consider a bad hexagon $P$ whose reflex vertices are separated by two convex vertices. Label the vertices $a$ through $f$ in cyclic order such that $c$ and $f$ are the reflex vertices, as shown in Figure~\ref{fig:mouse}.

\begin{figure}[ht]
    \centering
    \includegraphics[scale=0.5]{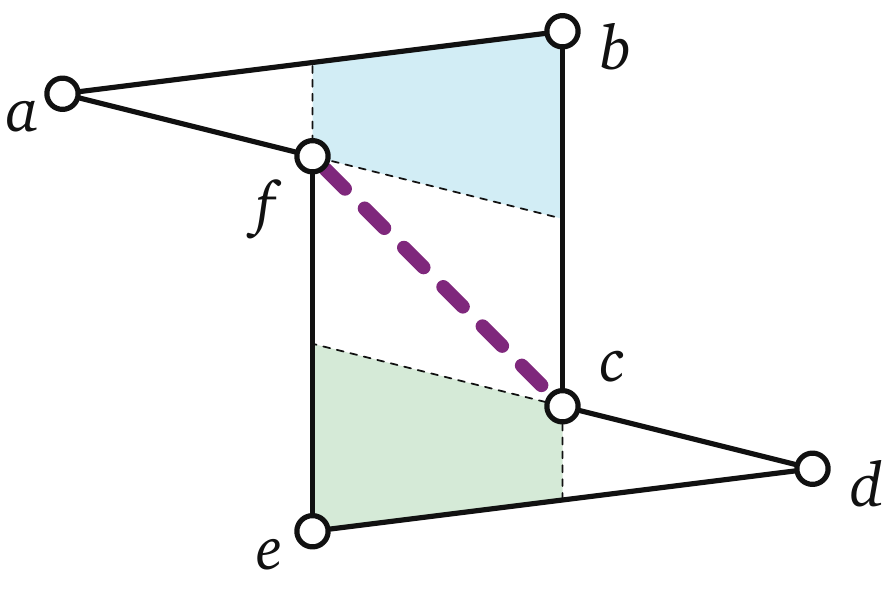}
    \caption{A hexagon with two reflex vertices that are separated by two convex vertices.  The diagonal between the reflex vertices splits the hexagon into two non-convex quadrilaterals with disjoint visibility kernels.}
    \label{fig:mouse}
\end{figure}

Because $P$ is bad, neither $c$ nor $f$ lies in its visibility kernel.  Vertex $c$ can see vertices $b$, $d$, and~$f$, so it cannot see both $a$ and $e$.  Without loss of generality, suppose $c$ cannot see $a$.  Then $f$ is a reflex vertex of the quadrilateral $abcf$.  Every quadrilateral has at most one reflex vertex, and that reflex vertex is in the visibility kernel, so $f$ can see $b$.  It follows that $f$ cannot see $d$, and $c$ is a reflex vertex of the quadrilateral $cdef$.

The visibility kernels of quadrilaterals $abcf$ and $cdef$ are disjoint, which implies that the visibility kernel of $P$ is empty.  We conclude that $P$ is not the link of any vertex.
\end{proof}

For the remainder of the proof, we annotate the edges of any triangulation as follows.  The \emph{star} of an edge in a triangulation is the union of the faces incident to that edge.  An edge is \EMPH{flippable} if its star is convex, and non-flippable otherwise.  Every non-flippable edge is incident to the unique reflex vertex of its star; we direct each non-flippable edge away from this reflex vertex.

\begin{lemma}\label{lem:green}
In any bad triangulation, every vertex is incident to exactly two incoming directed edges,
exactly two outgoing directed edges, and exactly two flippable edges.
\end{lemma}

\begin{proof}
Fix a bad triangulation $\Gam$. If $w$ is a reflex vertex in the link of some vertex $v$, then the edge $vw$ is
directed toward $v$. So Lemma~\ref{l:6regular2reflex} implies that each vertex is incident to two incoming edges. It follows
that the number of directed edges in $\Gam$ is exactly twice the number of vertices. Thus, each vertex is
incident to at most two outgoing edges. Because the number of directed edges in $G$ is exactly twice the
number of vertices, each vertex is incident to exactly two outgoing edges.
Finally, because every vertex of G is incident to exactly four directed edges, Lemma~\ref{l:6regular2reflex} implies that
every vertex of G is incident to two flippable edges.
\end{proof}

Lemmas~\ref{lem:nomice} and~\ref{lem:green} imply that the links in every bad triangulation fall into two categories, which we call \emph{cats} and \emph{dogs}.  A \EMPH{cat} is a bad hexagon whose reflex vertices are adjacent; a \EMPH{dog} is a bad hexagon whose reflex vertices are separated by one convex vertex. Cats are (combinatorially) symmetric; however, there are two species of dogs, which are reflections of each other.

\begin{figure}[ht]
\centering
\includegraphics[scale=0.5]{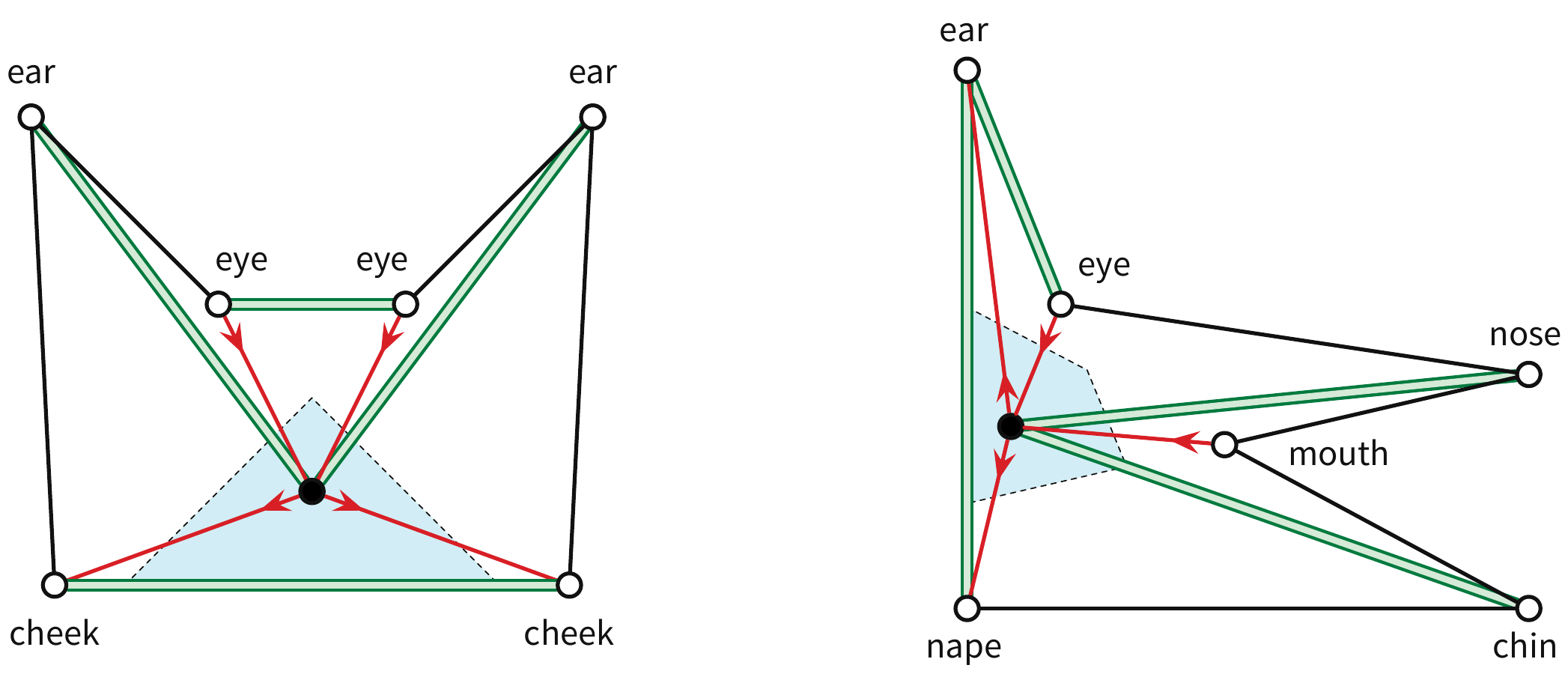}
\caption{Feline and (right-facing) canine anatomy. Flippable edges are shown in double-thick green.}
\label{F:catdog-anatomy}
\end{figure}

We label the vertices of each cat and dog mnemonically, as shown in Figure~\ref{F:catdog-anatomy}. For example, if the link of vertex $v$ is a dog, then $\emph{nose}(u)$ is the unique convex vertex of the dog whose neighbors are reflex, and the vertex opposite $\emph{nose}(u)$ in the link is $\emph{nape}(u)$. 
Since $\emph{nape}(u)$ is incident to two convex vertices, it can immediately see every other vertex in the link other than $\emph{nose}(u)$; its visibility of the nose is blocked by one of the two reflex vertices, which we will call $\emph{mouth}(u)$.  A dog is \EMPH{right-facing} if $\emph{mouth}(u)$ immediately follows $\emph{nose}(u)$ in clockwise order, and \EMPH{left-facing} otherwise.


The following lemma implies that two boundary edges of each cat and dog are also flippable, as shown in Figure~\ref{F:catdog-anatomy}.

\begin{lemma}
In any bad triangulation, every triangle is incident to exactly one flippable edge.
\end{lemma}

\begin{proof}
Fix a bad triangulation $\Gam$ with $n$ vertices.  Euler's formula implies that $\Gam$ has $3n$ edges and $2n$ triangular faces. Lemma~\ref{lem:green} implies that $\Gam$ has exactly $n$ flippable edges, so the average number of flippable edges per triangle is exactly $1$.

The flippable edges incident to any vertex $v$ are separated in cyclic order around $v$ by two non-flippable edges if the link of $v$ is a cat, or by one non-flippable edge if the link of $v$ is a dog.  Thus, two flippable edges never appear consecutively around a common vertex.  It follows that every triangle in $\Gam$ is incident to at most one flippable edge.
\end{proof}

\begin{lemma}
\label{L:not-all-dogs}
Every bad triangulation contains a cat.
\end{lemma}

\begin{proof}
Let $\Gam$ be a triangulation, and let $u$ be any vertex in $\Gam$ whose link is a dog.  We argue that the link of the nose of $u$ must be a cat.  (Mnemonically, dogs only sniff cats.)

Without loss of generality, assume that the link of $u$ is facing right, so the triple $\emph{chin}(u), \emph{nape}(u),\allowbreak \emph{ear}(u)$ is oriented clockwise, as shown in Figure \ref{F:dogangles}.  The fact that the link of $u$ is a \emph{bad} hexagon implies several orientation constraints on its vertices:
\begin{itemize}\itemsep2pt
\item
The triple $\emph{ear}(u), \emph{eye}(u), \emph{mouth}(u)$ is oriented counterclockwise; otherwise, $\emph{mouth}(u)$ could see the entire dog.
\item
The triple $\emph{eye}(u), \emph{mouth}(u), \emph{chin}(u)$ is oriented counterclockwise; otherwise, $\emph{eye}(u)$ could see the entire dog.
\item
Finally, the triple $\emph{nose}(u), \emph{mouth}(u), \emph{nape}(u)$ is oriented counterclockwise; otherwise, $\emph{nape}(u)$ could see the entire dog.
\end{itemize}

\begin{figure}[ht]
\centering
\includegraphics[scale=0.5]{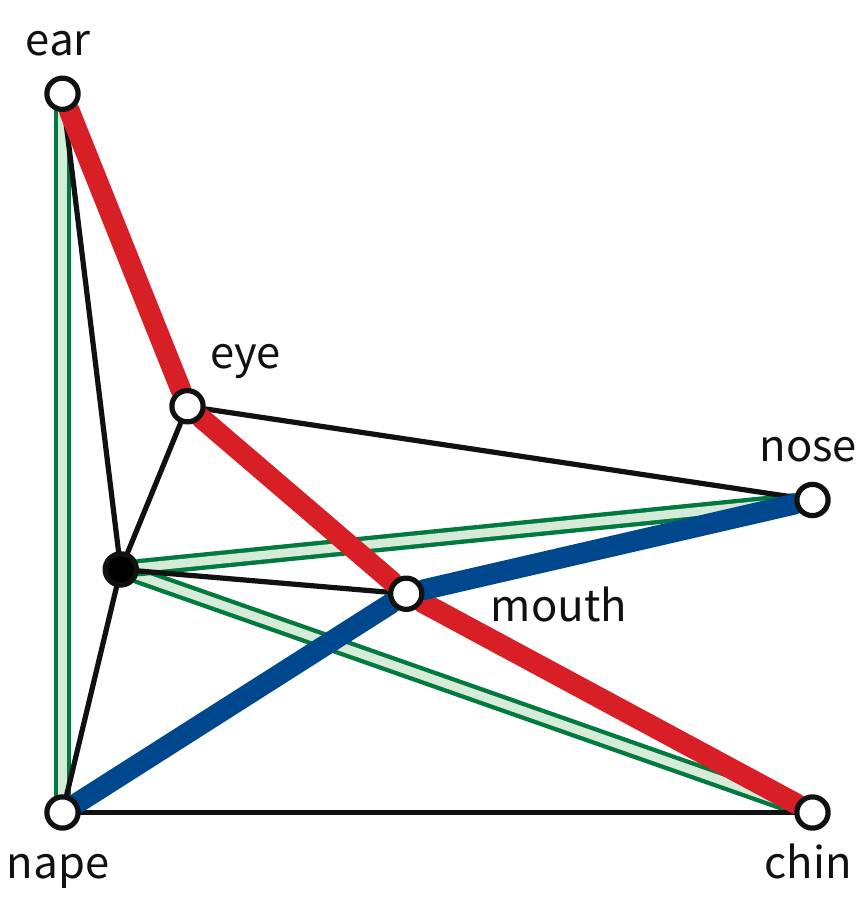}
\caption{Concave chains inside any dog.}
\label{F:dogangles}
\end{figure}

Now suppose for the sake of argument that the link of $v = \emph{nose}(u)$ is a dog.  We label the other vertices of the link of $u$ as shown on the top row of Figure \ref{F:got-no-nose}; in particular, $x = \emph{eye}(u)$ and $y = \emph{mouth}(u)$.  Because the flippable edge $uv$ is inside the link of $v$, either $u = \emph{nose}(v)$ or $u = \emph{chin}(v)$; each of these cases admits two subcases.  The four cases are illustrated schematically in the columns of Figure \ref{F:got-no-nose}.  

\begin{itemize}\itemsep3pt
\item
Suppose $u = \emph{nose}(v)$ and $x = \emph{mouth}(u)$, and therefore $y = \emph{eye}(u)$.  Let $w' = \emph{chin}(v)$.  The triple $w', x, y$ must be oriented clockwise; our earlier analysis implies that $w, x, y$ is oriented counterclockwise.  It follows that  triangles $w'xv$ and $wxu$ overlap, which is impossible.
\item
Suppose $u = \emph{nose}(v)$ and $x = \emph{eye}(u)$, and therefore $y = \emph{mouth}(u)$.  Let $w' = \emph{ear}(v)$.  The triple $w', x, y$ must be oriented clockwise; our earlier analysis implies that $w, x, y$ is oriented counterclockwise.  It follows that triangles $w'xv$ and $wxu$ overlap, which is impossible.
\item
Suppose $u = \emph{chin}(v)$ and $y = \emph{nape}(v)$.  Let $w' = \emph{nose}(v)$.  The triple $w', x, y$ is oriented clockwise; our earlier analysis implies that $w,y,z$ is oriented counterclockwise.  It follows that triangles $w'xv$ and $wxu$ overlap, which is impossible.
\item
Finally, suppose $u = \emph{chin}(v)$ and $y = \emph{mouth}(v)$.  Let $z' = \emph{nose}(v)$.  The triple $x, y, z'$ is oriented clockwise; our earlier analysis implies that $x,y,z$ is oriented counterclockwise.   It follows that triangles $z'yv$ and $zyu$ overlap, which is impossible.
\end{itemize}
In all four cases, we derive a contradiction.  We conclude that the link of $\emph{nose}(u)$ is actually a cat.
\end{proof}

\begin{figure}[ht]
\centering
\includegraphics[width=0.9\textwidth]{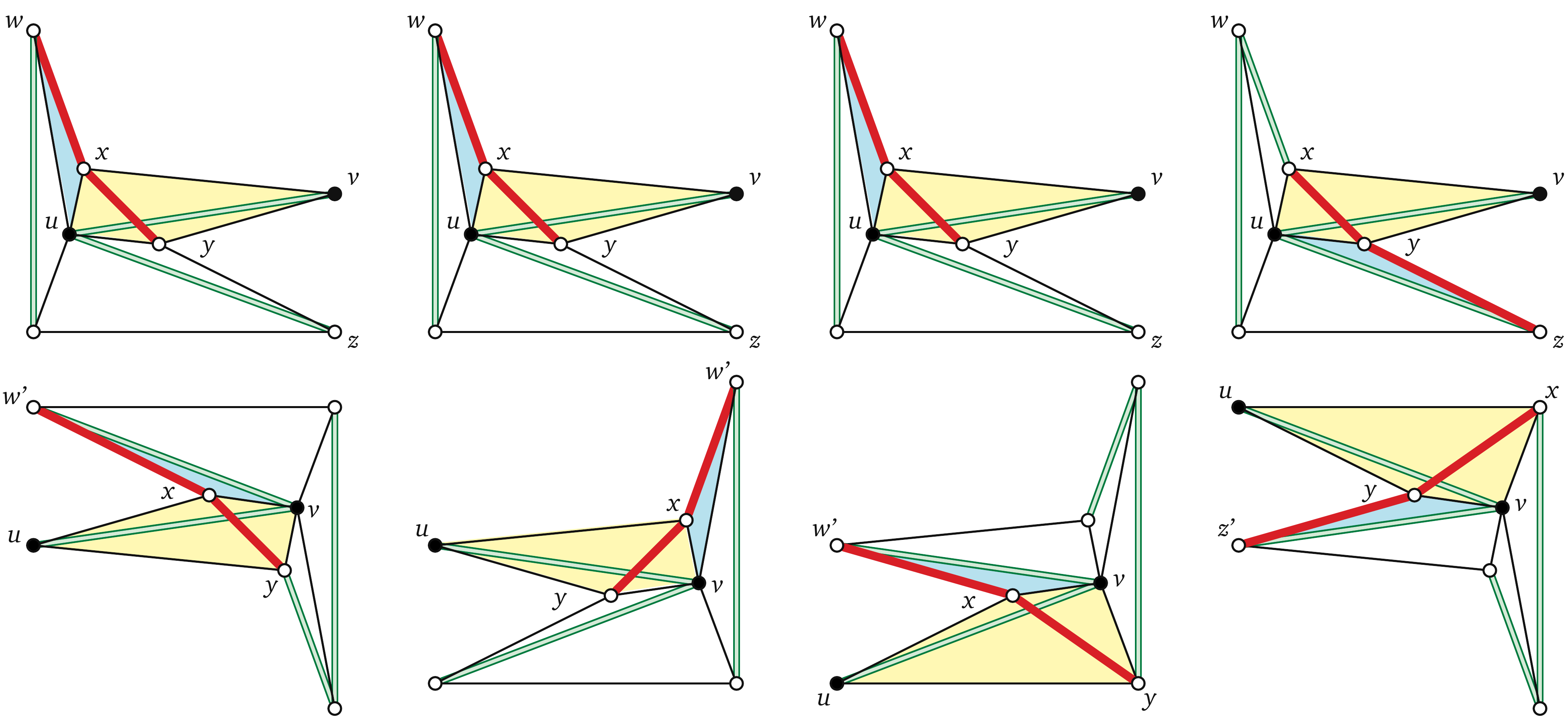}
\caption{A dog's nose ($v$) cannot be the center of another dog.}
\label{F:got-no-nose}
\end{figure}

Every 6-regular triangulation of the torus is isotopic to the quotient of the regular equilateral-triangle tiling of the plane by a lattice of translations \cite{a-cermt-73,n-ufetg-83,bk-emt-08}.  We analyze the patterns of cats and dogs in bad triangulations by analyzing their images in this reference triangulation, and in particular, by studying the induced annotations of edges, as illustrated in Figure \ref{F:catdog-reference}.

\begin{figure}[ht]
\centering
\includegraphics[scale=0.6]{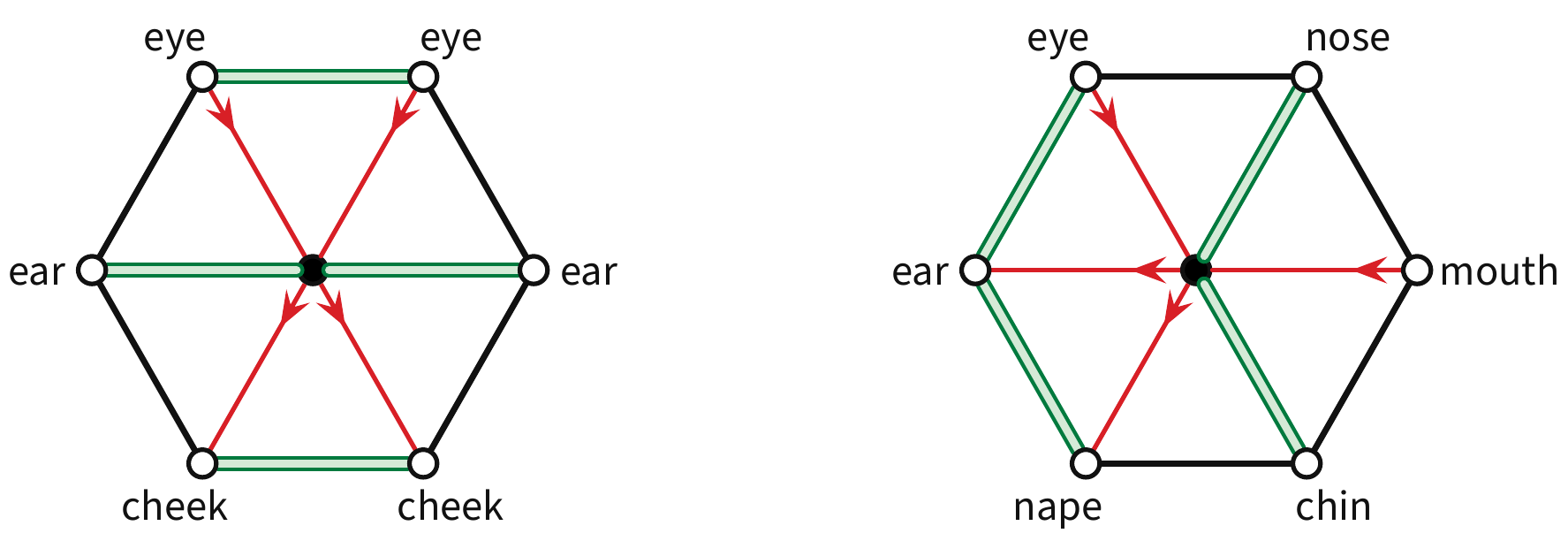}
\caption{Feline and (right-facing) canine reference anatomy; compare with Figure \ref{F:catdog-anatomy}.}
\label{F:catdog-reference}
\end{figure}

A cycle in a 6-regular toroidal triangulation is \EMPH{straight} if it corresponds to a closed geodesic in the corresponding lattice triangulation.  Every straight cycle is non-contractible.

Let $C$ be a simple polygonal (i.e., piecewise geodesic) cycle, arbitrarily directed, and let $v$ be any vertex of $C$. Let $\widetilde{v}$ be any lift of $v$, and let $\widetilde{u}\widetilde{v}$ and $\widetilde{v}\widetilde{w}$ be lifts of the edges of $C$ incident to $v$. The \EMPH{turning angle} of $C$ at $u$ is the signed counterclockwise angle, strictly between $-\pi$ and $\pi$, between the vectors $\arc{\widetilde{u}}{\widetilde{v}}$ and $\arc{\widetilde{v}}{\widetilde{w}}$.  In other words, when walking along the cycle $C$ in the indicated direction, the turning angle is the angle one turns left at~$v$ (or right if the angle is negative).  The \EMPH{total turning angle} of $C$ is the sum of the turning angles of the vertices of $C$.

\begin{lemma}
\label{L:turn}
Every simple non-contractible cycle on the flat torus has total turning angle zero.
\end{lemma}

Lemma \ref{L:turn} follows immediately from classical results of Reinhart \cite{r-let-59}.  In short, the total turning angle (which Reinhart calls the “winding number”) is an isotopy invariant of closed curves, and every simple non-contractible cycle on the flat torus is isotopic to a closed geodesic.  Lemma \ref{L:turn} implies in particular that every \emph{straight} cycle has total turning angle zero.

\begin{lemma}
If a bad triangulation contains one dog, it contains a straight cycle of dogs.
\end{lemma}

\begin{proof}
Let $\Gam$ be a bad triangulation and let  $v_1$ be any vertex of $\Gam$ whose link is a dog.  Label the vertices of $v$'s link as shown in Figure \ref{F:dogchain}; for example, $u_1 = \emph{nose}(v_1)$, $v_2 = \emph{mouth}(v_1)$, and $w_1 = \emph{chin}(v_1)$.

Edges $u_1v_1$ and $v_1w_1$ are flippable, and therefore edges $v_1v_2$, $u_1v_2$, and $w_1v_2$ are not flippable.  Thus, no opposite pair of edges incident to $v_2$ are both flippable.  It follows that the link of $v_2$ is a dog, whose ear vertex is~$v_1$.

Edges $u_0v_1$, $v_0v_1$, and $v_1w_0$ are not flippable, so edges $u_0v_0$ and $v_0w_0$ must be flippable.  Thus, the link of $v_0$ is also a dog, whose mouth vertex is $v_1$.

\begin{figure}[ht]
\centering
\includegraphics[scale=0.5]{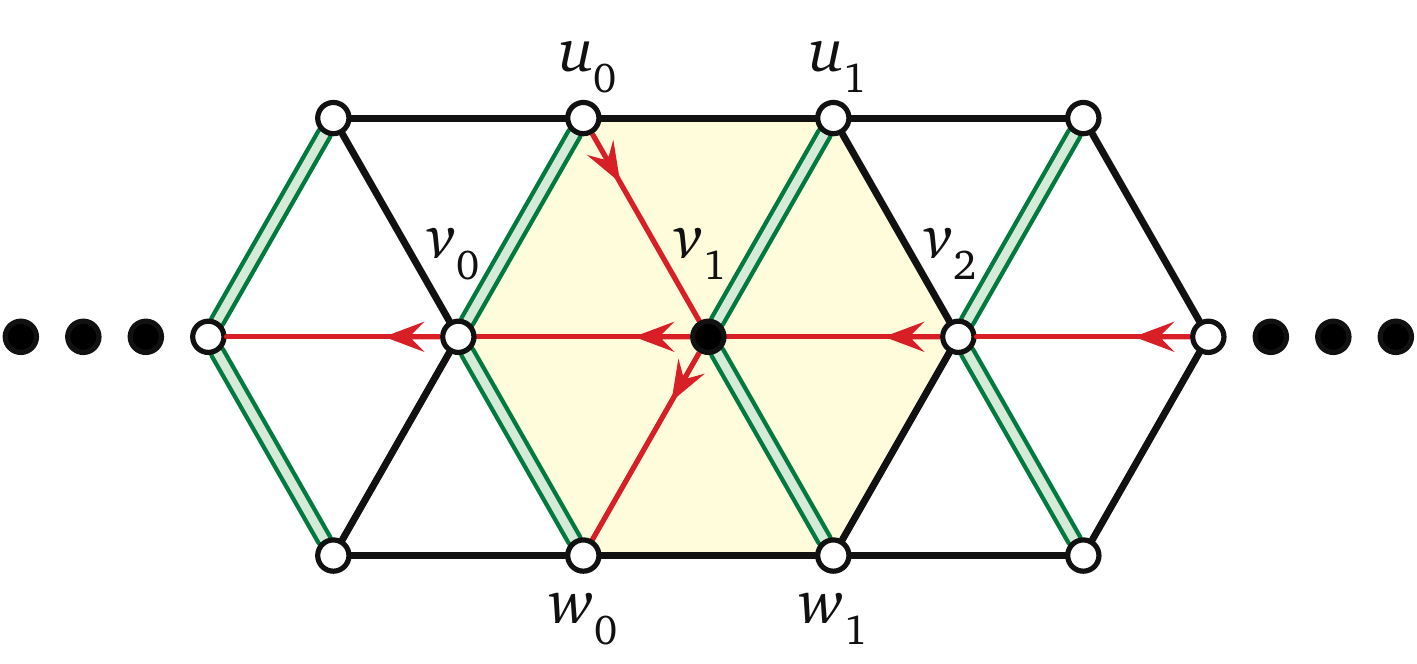}
\caption{One dog induces a straight cycle of dogs.}
\label{F:dogchain}
\end{figure}

Continuing by induction in both directions, we find a bidirectional sequence $\dots, v_{-1}, v_0, v_1, v_2, \dots$ of vertices whose links are dogs, where the center $v_i$ of each dog is the mouth of the previous dog and the ear of the next dog.  Because $\Gam$ is finite, this sequence must eventually repeat, forming a straight cycle.
\end{proof}

Finally, we are ready to prove the main theorem of this section.

\begin{theorem}\label{thm:nobadtriangulations}
Bad triangulations do not exist.
\end{theorem}

\begin{proof}
Let $\Gam$ be a bad triangulation.  We derive a contradiction by showing that $\Gam$ contains a non-contractible cycle whose vertices all have cat links, whose edges are all non-flippable, and finally whose turning angle is non-zero, contradicting Lemma \ref{L:turn}.

Let $u_1$ be any vertex of $\Gam$ whose link is a cat; Lemma \ref{L:not-all-dogs} implies that such a vertex exists.  Label the vertices in the neighborhood of $u_1$ as shown in Figure \ref{F:catchain}.  The directions of the non-flippable edges of the cat are not relevant, so we will ignore them.  Thus, for example, $u_2$ and~$v_2$ are either the eyes of $u_1$'s link or its cheeks.

First, suppose that $\Gamma$ contains at least one dog, and therefore without loss of generality that $u_1$ is adjacent to a vertex $v_1$ whose link is a dog.  Up to symmetry, there are two cases to consider.  If the link of~$v_1$ is a dog whose mouth is $v_2$, then the link of $v_2$ is also a dog (whose ear is~$v_1$).  Conversely, if the link of $v_2$ is a dog, the ear of that dog is $v_1$, and thus, the link of $v_1$ is also a dog whose mouth is $v_2$.  The previous lemma now implies a straight cycle of dogs $D = \dots, v_0, v_1, v_2, v_3, \dots$.

\begin{figure}[ht]
\centering
\includegraphics[scale=0.5]{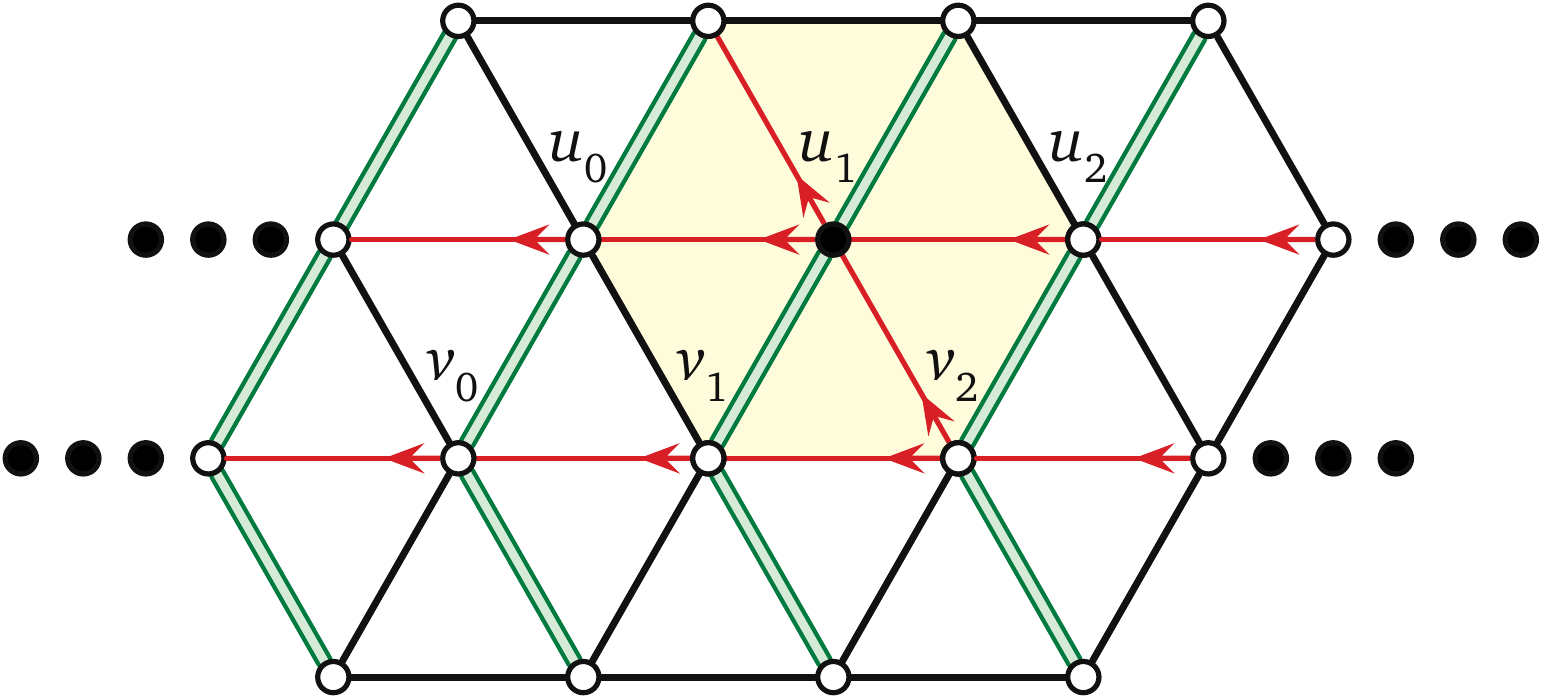}
\caption{One cat induces a straight cycle of cats.}
\label{F:catchain}
\end{figure}

Because the link of $v_1$ is a dog, edge $u_0v_0$ is flippable; it follows that the link of $u_0$ is a cat.  Similarly, because the link of $v_3$ is a dog, edge $u_3v_3$ is flippable; it follows that the link of $u_2$ is a cat.  Proceeding inductively in both directions, we find a straight cycle $C = \dots, u_0, u_1, u_2, u_3, \dots$ of vertices, parallel to the cycle $D$ of dogs, all of whose links are cats.  Every edge $u_i u_{i+1}$ in $C$ is unflippable.

On the other hand, if $\Gam$ contains no dogs, then we can construct a straight cycle $C = \dots, u_0, u_1, u_2, \dots$ of cat vertices starting with any unflippable edge $u_0u_1$.  Again, every edge $u_iu_{i+1}$ in $C$ is unflippable.

Now suppose $u_2$ is the left ear of $u_1$.  Then by induction, for every index~$i$, vertex $u_{i+1}$ is the left ear of $u_i$ and $u_{i-1}$ is the right cheek of $u_i$.  It follows that every triple $u_{i-1} u_i u_{i+1}$ is oriented clockwise, and thus the total turning angle of $C$ is negative, contradicting Lemma \ref{L:turn}.  The other three cases  similarly lead to contradictions; in all cases, every triple $u_{i-1} u_i u_{i+1}$ has the same orientation, which implies that the total turning angle of $C$ is non-zero, contradicting Lemma~\ref{L:turn}.
\end{proof}

\section{Equilibrium Pseudomorphs}
\label{sec:eqpseudo}

\def\Tutte#1#2{\textit{Eq}(#1, #2)}
\def\Unit{\mathbf{e}}

As pointed out in Section~\ref{sec:pseudomorph}, when computing a recursive pseudomorph between two triangulations~$\Gam_0$ and~$\Gam_1$, we must contend with the fact that while each triangulation must contain a good vertex, there may be no way to safely collapse the same vertex along the same edge in both triangulations.  It is relatively straightforward to adapt Cairns' approach for planar graphs to the torus, by recursively computing pseudomorphs from both $\Gam_0$ and $\Gam_1$ to a suitable intermediate triangulation; unfortunately, the resulting pseudomorph from $\Gam_0$ to $\Gam_1$ consists of an exponential number of steps.

Instead, our algorithm introduces \emph{one} intermediate triangulation isotopic to both $\Gam_0$ and $\Gam_1$, namely an equilibrium triangulation $\Gam_*$ guaranteed by Y.~Colin de Verdière's generalization of Tutte's spring embedding theorem (Theorem~\ref{Th:tutte-torus}).  In this section, we describe how to collapse an \emph{arbitrary} edge $e$ in an equilibrium triangulation to obtain a simpler equilibrium triangulation.  This operation allows us to recursively compute a pseudomorph from either $\Gam_0$ or $\Gam_1$ to $\Gam_*$ using only a linear number of steps, as described in Section~\ref{ssec:triangulationpseudomorph}. 
 
Intuitively, we continuously increase the weight of $e$ to $\infty$ and maintain the equilibrium triangulation.  We show that the resulting “spring collapse” moves all vertices of the equilibrium embedding along geodesics parallel to $e$, as shown in Figure~\ref{F:Tutte-parallel}; in particular, edge $e$ collapses to a single vertex.  It follows that straightforward linear interpolation from one equilibrium triangulation to the other is a parallel linear pseudomorph.

\begin{figure}[ht]
\centering
\includegraphics[scale=0.5]{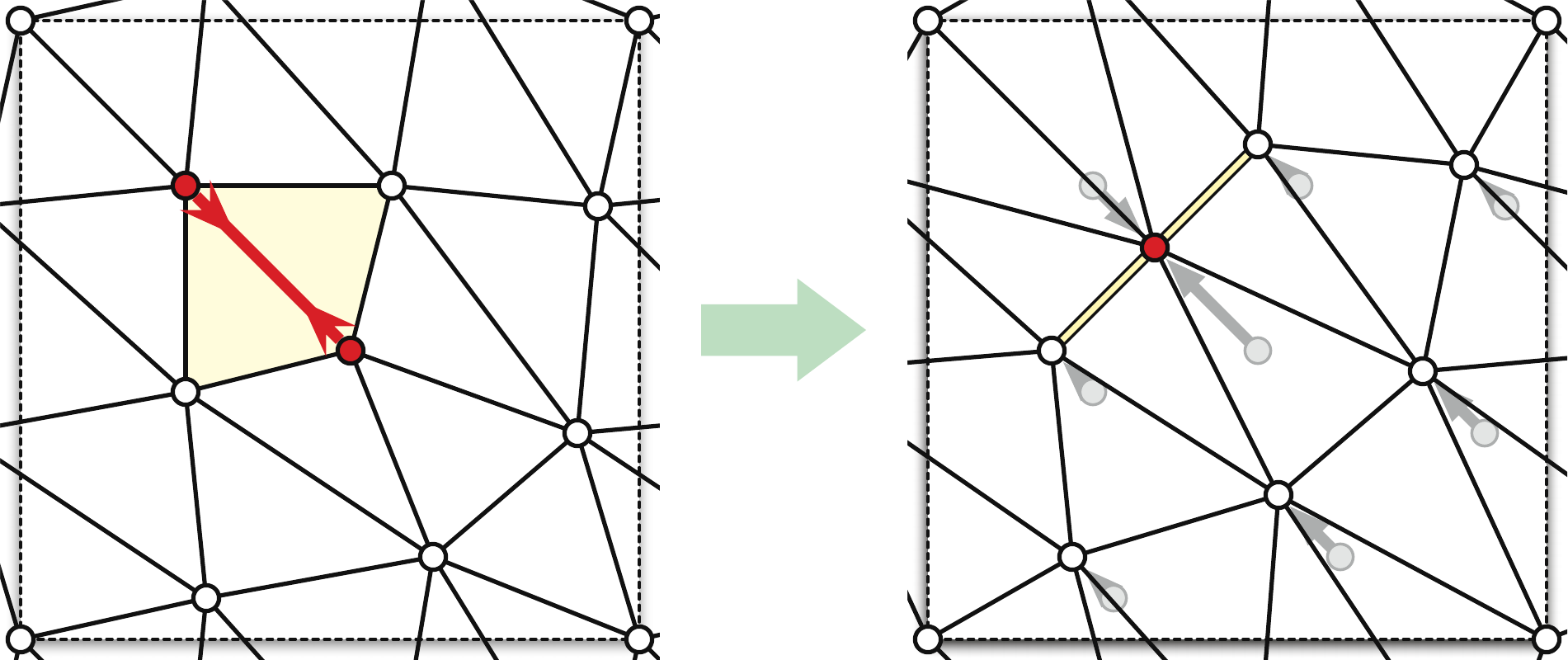}
\caption{Increasing the weight of any edge $e$ to infinity collapses $e$ and moves every vertex parallel to $e$.}
\label{F:Tutte-parallel}
\end{figure}

For any triangulation $\Gamma$ and any assignment of positive weights $\lambda(e)$ to the edges of $\Gamma$, let \EMPH{$\Tutte{\Gamma}{\lambda}$} denote the geodesic embedding isotopic to~$\Gamma$ that is in equilibrium with respect to the weight vector~$\lambda$, as guaranteed by Theorem~\ref{Th:tutte-torus}.  (At the top level of recursion, we can safely assume that $\lambda(e)=1$ for every edge $e$, but as we point out in Section \ref{SS:spring-collapse}, we cannot maintain that assumption  in deeper recursive calls.)  For each vertex $v$ of $\Gamma$, let \EMPH{$p_*(v,\lambda)$} denote the coordinates of $v$ in $\Tutte{\Gamma}{\lambda}$, obtained by solving linear system \eqref{Eq:Tutte}.  To define the embedding~$\Tutte{\Gamma}{\lambda}$ and its coordinates $p_*(v,\lambda)$ uniquely, fix $p_*(o, \lambda) = (0,0)$ for some arbitrary vertex~$o$.

\subsection{Parallel Motion}
\label{SS:parallel}

The following lemma states intuitively that changing a single edge weight $\lambda(e)$ moves each vertex of the equilibrium embedding $\Tutte{\Gamma}{\lambda}$ along a geodesic parallel to $e$.

\begin{lemma}
\label{L:tutte-parallel}
Let $\lambda$ and $\lambda'$ be arbitrary positive edge weights such that $\lambda(e) \ne \lambda'(e)$ for some edge $e$, and $\lambda(e') = \lambda'(e')$ for all edges $e' \ne e$. Let $d$ be a dart of $e$ with tail $u$ and head $v$. Then for every vertex $w$, the vector $p_*(w,\lambda') - p_*(w,\lambda)$ is parallel to $p_*(v,\lambda) - p_*(u,\lambda) + x(d)$.
\end{lemma}

\begin{proof}
Arbitrarily index the vertices of $\Gamma$ from $1$ to $n$.  Fix a non-zero vector $\sigma \in \Real^2$ orthogonal to the vector $p_*(v,\lambda) - p_*(u,\lambda) + x(d)$.  For each vertex $i$, let $z_i = p_*(i, \lambda) \cdot \sigma$ and $z'_i= p_*(i, \lambda') \cdot \sigma$, and for each dart $d$, let $\chi(d) = x(d) \cdot \sigma$.  Our choice of $\sigma$ implies that $z_u - z_v = \chi(d)$.  We need to prove that $z'_i = z_i$ for every vertex $i$.

The real vector $Z = (z_i)_i$ is a solution to the linear system $LZ = X$, where $L$ is the $n\times n$ weighted Laplacian matrix
\[
	L_{ij} = \begin{cases}
		\displaystyle\sum_{\Tail(d) = i} - \lambda(d)
				& \text{if $i = j$} 
		\\[3ex]
		\displaystyle\sum_{\substack{\Tail(d) = i \\ \Head(d) = j}} \lambda(d)
				& \text{if $i\ne j$}
	\end{cases}
\]
and $H \in \Real^n$ is a vector whose $i$th entry is 
\[
	H_i = \displaystyle\sum_{\Tail(d) = i} - \lambda(d)\,\chi(d).
\]
(In fact, $Z$ is the \emph{unique} solution such that $z_o = 0$.)  
Similarly, $Z' = (z'_i)_i$ is the unique solution to an analogous equation $L' Z' = H'$ with $z'_0 = 0$, where $L'$ and $H'$ are defined \emph{mutatis mutandis} in terms of~$\lambda'$ instead of $\lambda$.

We prove that $Z' = Z$ as follows.  Let $\delta = \lambda'(e) - \lambda(e)$.  The Laplacian matrices $L$ and $L'$ differ in only four locations:
\[
	L'_{ij} - L_{ij} = \begin{cases}
		-\delta & \text{if $i=j=u$ or $i=j=v$} \\
		\delta & \text{if $\set{i,j} = \set{u,v}$} \\
		0 & \text{otherwise}
	\end{cases}
\]
More concisely, we have $L' = L - \delta\, (\Unit_v - \Unit_u)\,(\Unit_v - \Unit_u)^T$, where $\Unit_i$ denotes the $i$th standard coordinate vector.  Similarly, we have $H' = H + \delta\cdot \chi(d)\cdot (\Unit_v - \Unit_u)$.  It follows that 
\begin{align*}
	L' Z
	&=
	L Z - \delta\,(\Unit_v - \Unit_u)\,(\Unit_v - \Unit_u)^T\, Z 
\\	&=
	H - \delta\,(\Unit_v - \Unit_u)\,(z_v - z_u)
\\	&=
	H + \delta\,(\Unit_v - \Unit_u)\, \chi(d)
\\	&=
	H'
\end{align*}
which completes the proof.
\end{proof}

We note in passing that a nearly identical lemma applies to internally 3-connected plane graphs; changing the weight of one edge $e$ moves the vertices of any Tutte embedding along lines parallel to $e$.  Surprisingly, this observation appears to be new. 

\subsection{Spring Collapse}
\label{SS:spring-collapse}

Now fix a toroidal triangulation $\Gamma$ and edge weights $\lambda$.  Let $e$ be an arbitrary edge of $\Gamma$, and let $\Gamma'$ be the result of collapsing one endpoint $u$ to the other endpoint $v$.  Let $a,b,c,d$ be the edges in the link of $uv$ in~$\Gamma$, as shown in Figure \ref{F:collapse}.  Edges $a$ and $b$ collapse to a single edge $ab$ in $\Gamma'$, and edges $c$ and $d$ collapse to a single edge $cd$ in $\Gamma'$.  Now define weights for the edges of $\Gamma'$ as follows:
\[
	\lambda'(e) :=
		\begin{cases}
			\lambda(a) + \lambda(b) 	& \text{if $e = ab$} \\
			\lambda(c) + \lambda(d) 	& \text{if $e = cd$} \\
			\lambda(e)				 	& \text{otherwise}
		\end{cases}
\]
\begin{figure}[ht]
\centering
\includegraphics[scale=0.4]{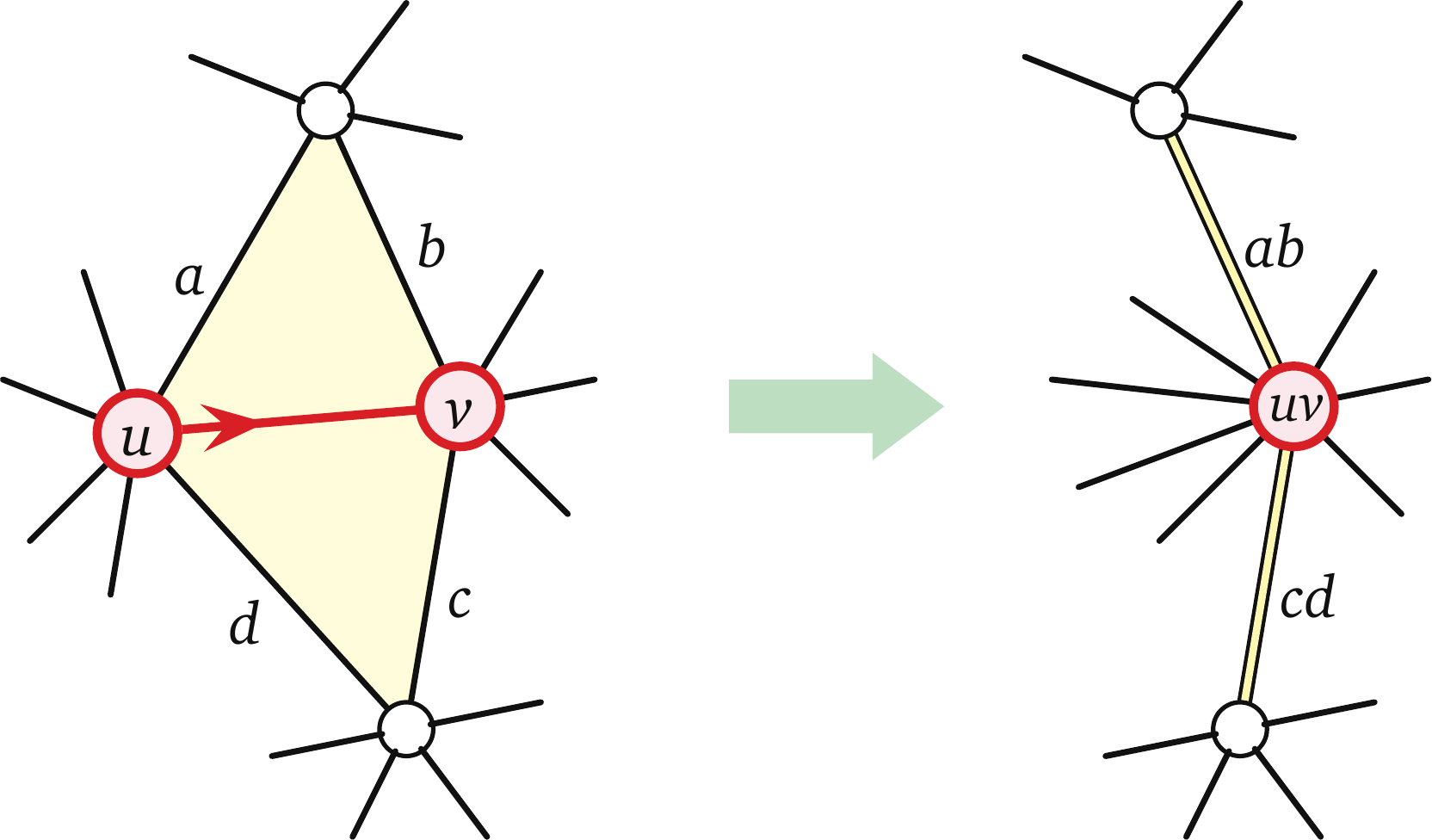}
\caption{Collapsing $u$ to $v$.}
\label{F:collapse}
\end{figure}

The new equilibrium embedding $\Tutte{\Gamma'}{\lambda'}$ has the same image as the limit of $\Tutte{\Gamma}{\lambda}$ as we increase $\lambda(e)$ to infinity and keep all other edge weights fixed.  Lemma \ref{L:tutte-parallel} implies that as we continuously increase $\lambda(e)$, the vertices of $\Tutte{\Gamma}{\lambda}$ move continuously along geodesics parallel to $e$.

\subsection{Geodesic Interpolation}
\label{SS:interp}

This continuous deformation from $\Tutte{\Gamma}{\lambda}$ to $\Tutte{\Gamma'}{\lambda'}$ is not a parallel linear pseudomorph, because the vertices do not necessarily move at fixed speeds.\footnote{We have deliberately left unspecified exactly \emph{how} $\lambda(e)$ grows to infinity.  We conjecture that with the right choice of growth function, the resulting vertex motion is actually linear.}  To define a parallel linear pseudomorph, we simply move each vertex at constant speed along its corresponding geodesic.  That is, for every vertex $w$ and all real numbers $0\le t\le 1$, let
\[
	p_t(w) = (1-t)\cdot p_*(w, \lambda) + t\cdot p_*(w, \lambda'),
\]
and let $\Gamma_t$ denote the geodesic drawing of $G$ with vertex coordinates $p_t(v)$ and the same crossing vectors as $\Gamma$.  The following lemma implies that $\Gamma_t$ is actually an embedding for all $t<1$, which implies that the continuous family of drawings $\Gamma_t$ is a parallel linear pseudomorph from  $\Gamma_0 = \Tutte{\Gamma}{\lambda}$ to $\Gamma_1 = \Tutte{\Gamma'}{\lambda'}$.

\begin{lemma}\label{lem:paralleltriangles}
Let $p_0p_1$, $q_0q_1$, and $r_0r_1$ be arbitrary parallel segments in the plane.  For all real $0\le t\le 1$, define $p_t = (1-t)p_0 + tp_1$ and $q_t = {(1-t)}q_0 + tq_1$ and $r_t = (1-t)r_0 + tr_1$.  If the triples $p_0,q_0,r_0$ and $p_1,q_1,r_1$ are oriented counterclockwise, then for all $0\le t\le 1$, the triple $p_t,q_t,r_t$ is also oriented counterclockwise.
\end{lemma}

\begin{proof}
Without loss of generality, assume that segments $p_0p_1$, $q_0q_1$, and $r_0r_1$ are horizontal.  Thus, we can write $p_t = (px_t, py)$, and similarly for $q_t$ and $r_t$.  The triple $p_t,q_t,r_t$ is oriented counterclockwise if and only if the following determinant is positive:
\[
	\Delta(t) := 
	\begin{vmatrix}
		1 & px_t & py \\
		1 & qx_t & qy \\
		1 & rx_t & ry 
	\end{vmatrix}
\]
Routine calculation implies
\begin{align*}
	\Delta(t)
	&=
	\begin{vmatrix}
		1 & (1-t)\,px_0 + t\,px_1 & py \\
		1 & (1-t)\,qx_0 + t\,qx_1 & qy \\
		1 & (1-t)\,rx_0 + t\,rx_1 & ry 
	\end{vmatrix}
\\	&=
	(1-t)
	\begin{vmatrix}
		1 & px_0 & py \\
		1 & qx_0 & qy \\
		1 & rx_0 & ry 
	\end{vmatrix}
	+ 
	t
	\begin{vmatrix}
		1 & px_1 & py \\
		1 & qx_1 & qy \\
		1 & rx_1 & ry 
	\end{vmatrix}
	~=~
	(1-t)\cdot\Delta(0)
	+ t\cdot \Delta(1)
\end{align*}
Thus, the function $\Delta(t)$ has exactly one real root.  It follows that if $\Delta(0) > 0$ and $\Delta(1) > 0$, then  $\Delta(t)>0$ for all $0\le t\le 1$.
\end{proof}


\section{Zippers}
\label{sec:zippers}

Even if the original input triangulations $\Gam_0$ and $\Gam_1$ are simple, collapsing edges eventually reduces them to triangulations with parallel edges and loops.  Every loop in a geodesic toroidal triangulation is a closed geodesic.  The base case of our recursive algorithm is a special type of geodesic toroidal triangulation that we call a \EMPH{zipper}, in which \emph{every} vertex is incident to a loop.  (This class of graphs were previously considered by Gonçalves and Lévêque \cite[Fig.~44]{gl-tmswo-14}.)

Consider any zipper $Z$ with $n$ vertices, for some positive integer $n$.  If $n=1$, then $Z$ consists of a single vertex, three loop edges, and two triangular faces.  Otherwise, the loops in $Z$ are disjoint closed geodesics, so they must be parallel; it follows that each vertex of $Z$ is incident to exactly one loop.  In either case, the $n$ loops in $Z$ decompose the torus into $n$ annuli, each of which is decomposed into two triangles by two boundary-to-boundary edges.  Figure~\ref{fig:two-vertices-torus} shows three two-vertex zippers, and Figure \ref{F:5zipper} shows a zipper with five vertices.

\begin{figure}[ht]
\centering\includegraphics[scale=0.5]{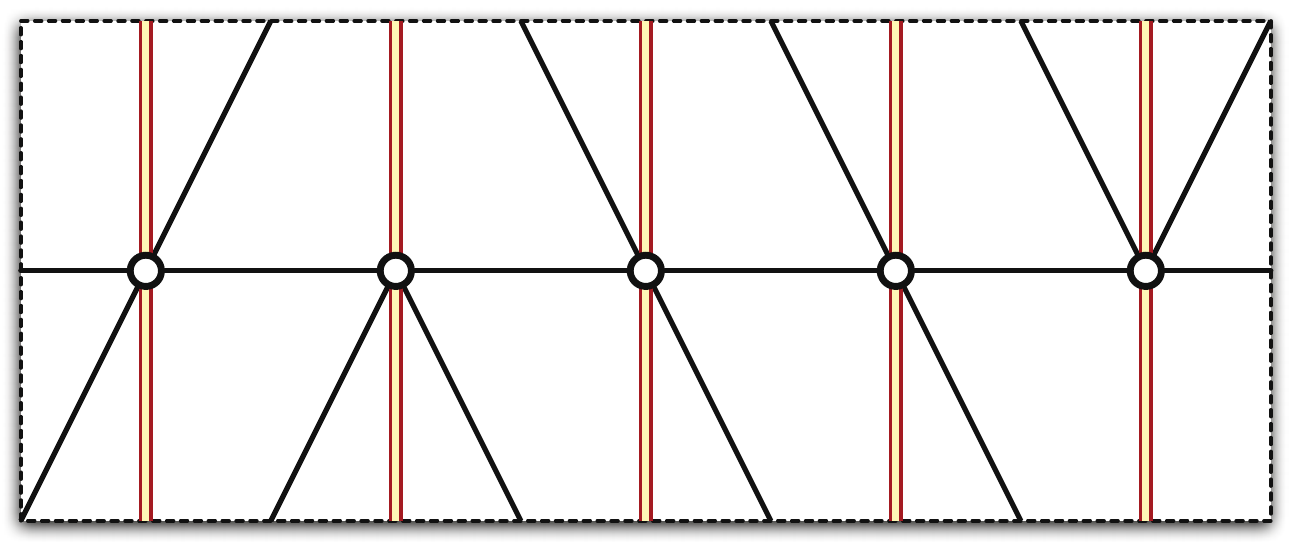}
\caption{A five-vertex zipper.  Doubled red edges are loops.}
\label{F:5zipper}
\end{figure}

\subsection{Zipper Structure}

The following results motivate our choice of zippers as a base case of our recursive algorithm.

\begin{lemma}
\label{L:loop6}
In every geodesic toroidal triangulation, every vertex incident to a loop has degree at least~6.
\end{lemma}

\begin{proof}
Let $\Gamma$ be a geodesic toroidal triangulation, let $v$ be a vertex of $\Gamma$ incident to a loop, let $d_0$ be either of the darts of that loop.

Let $\partial f$ denote the clockwise facial walk around the face $f$ to the right of~$d_0$.  Because the interior of~$f$ is an open disk, $\partial f$ is contractible.  Because $\Gamma$ is a triangulation, $\partial f$ consists of exactly three darts $d_0$,~$d_1$, and $d_2$, where $\emph{head}(d_i) = \emph{tail}(d_{i+1\bmod 3})$ for each index $i$.  Because every contractible geodesic loop consists of a single point, $d_0$ is non-contractible and thus is not homotopic to $\partial f$.  It follows that $d_2 \ne \emph{rev}(d_1)$.

Symmetrically, the counterclockwise walk around the face to the left of $d_0$ consists of three darts $d_0, d'_1, d'_2$, where $d'_2 \ne \emph{rev}(d'_1)$.  Thus, at least six distinct darts head into~$v$: in counterclockwise cyclic order, $d_0$, $\emph{rev}(d_1)$, $d_2$, $\emph{rev}(d_0)$, $d'_2$, $\emph{rev}(d'_1)$.\footnote{If $\Gamma$ has only one vertex, then $d'_2 = d_1$ and $d'_1 = d_2$, and so $v$ is incident to only three distinct \emph{edges}.  Otherwise, $v$ is incident to five distinct edges.}
\end{proof}

\begin{figure}[ht]
\centering\footnotesize\sffamily
\begin{tabular}{c@{\hspace{1in}}c}
	\includegraphics[scale=0.4,page=1]{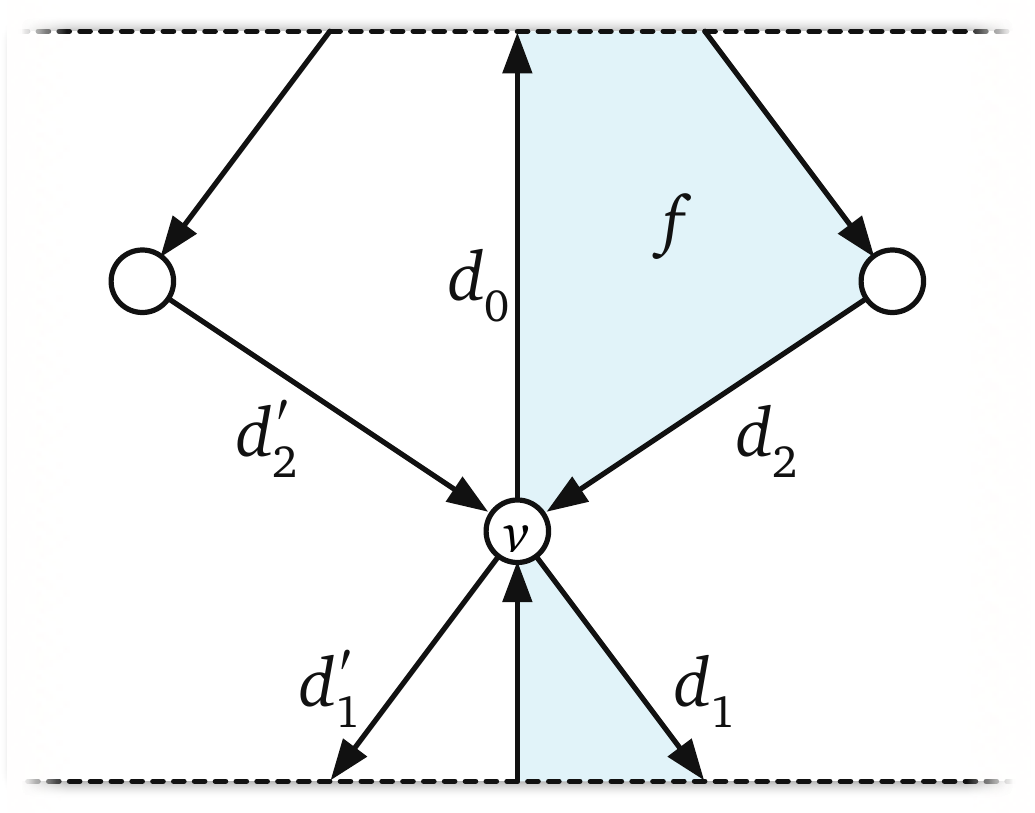} &
	\includegraphics[scale=0.4,page=2]{Fig/loop6} \\ (a) & (b)
\end{tabular}
\caption{(a) Proof of Lemma \ref{L:loop6}.
		  (b) Proof of Lemma \ref{L:loopmeanszipper}}
\end{figure}

\begin{lemma}
\label{L:loopmeanszipper}
Every 6-regular triangulation that contains a loop is a zipper.
\end{lemma}

\begin{proof}
Again, let $\Gamma$ be a 6-regular triangulation, and let $v$ be a vertex of $\Gamma$ incident to a loop $\ell$.  The previous proof implies that $v$ is incident to two edges on either side of $\ell$.  Let $e$ and $e'$ be the edges incident to $v$ on one side of $\ell$; the edges $\ell$, $e$, and $e'$ enclose a triangular face $f$.  Thus, $e$ and $e'$ share another common endpoint $w$.  (Except in the trivial case where $\Gamma$ has only one vertex, $v$ and $w$ are distinct.)  Because $e$ and $e'$ are adjacent in cyclic order around $v$, there is another triangular face $f'$ with $e$ and $e'$ on its boundary; the third edge of $f'$ is a loop through $w$.  The lemma now follows by induction.
\end{proof}

\begin{corollary}
Every triangulation that contains a loop but is not a zipper contains a vertex of degree at most~5 that is not incident to a loop.
\end{corollary}

\subsection{Morphing Zippers}

We next describe a straightforward approach to morphing between arbitrary isotopic zippers, which requires at most two parallel linear morphing steps.

Let $Z$ and $Z'$ be arbitrary isotopic zippers with $n$ vertices.  If $n=1$, then $Z$ and $Z'$ differ only by translation, so assume otherwise.  Ladegaillerie \cite{l-ctp1c-74a,l-ctp1c-74b,l-cdp1c-84} proved that two embeddings of the same graph~$G$ on the same surface are isotopic if and only if, the images of any cycle in $G$ in both embeddings are homotopic.  Ladegaillerie's theorem implies that all cycles in $Z$ and all cycles in $Z'$ are parallel to a common vector $\sigma$.

In the first parallel morphing step, we translate all vertices in $Z$ along geodesics orthogonal to $\sigma$ until the image of each loop has the same image as the corresponding loop in $Z'$.  Then in the second parallel morphing step, we translate all vertices along their respective loops to move all vertices and edges to their proper positions in $Z'$.  See Figure~\ref{F:zipper-morph} for an example.  In both stages, Lemma \ref{lem:paralleltriangles} implies that linear interpolation between the old and new vertex coordinates yields an isotopy.

\begin{figure}[ht]
\centering
\includegraphics[scale=0.5]{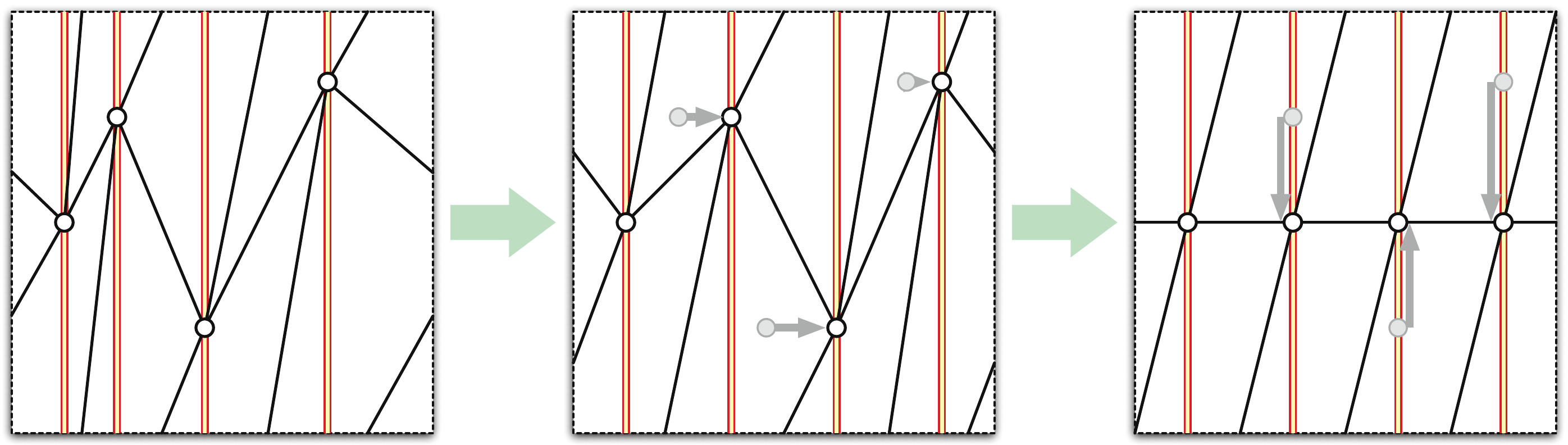}
\caption{Morphing one zipper into another.  Doubled red edges are loops.}
\label{F:zipper-morph}
\end{figure}

\section{Converting Pseudomorphs to Morphs}
\label{sec:depseudo}


In this section, we adapt a perturbation strategy of Alamdari et al.~\cite{aabcd-hmpgd-17}, which transforms their pseudomorphs between planar triangulations into morphs, to geodesic triangulations on the flat torus.


To explain our adaptation, we must first give a brief sketch of their algorithm.  Let $\Gam_0$ be the initial planar input triangulation.  The input to their perturbation algorithm is a pseudomorph consisting of a (direct) collapse of a good vertex $u$ to a neighbor $v$, a morph (not a pseudomorph) consisting of $k$ parallel linear steps $\Gam'_0 \leadsto \Gam'_1 \leadsto\cdots \leadsto \Gam'_k$, and finally the reverse of a collapse of $u$ to $v$ resulting in the final triangulation~$\Gam_k$.  The output is a proper morph from $\Gam_0$ to $\Gam_k$, consisting of $k+2$ parallel linear steps.

Alamdari et al.~proceed as follows.  Let $P$ be the link of $u$ in the initial triangulation~$\Gam_0$.  For each index~$i$, let $v_i$ and $P_i$ respectively denote the images of $v$ and $P$ in the intermediate triangulation $\Gam'_i$.  For each index $i$, a position $u_i$ is found within the visibility kernel of $P_i$ so that for all $i$, the vector $u_i \arcto u_{i+1}$ is parallel to $v_i\arcto v_{i+1}$ (the direction of the parallel linear morph $\Gam'_i \leadsto \Gam'_{i+1}$).  For vertices of degree $3$ and~$4$, it is simple to place $u$ as a certain convex combination of the vertices in $P$.  The strategy for vertices of degree $5$ is more complicated.  First, a value $\eps$ is computed such that for each index $i$, the intersection of the disk of radius $\eps$ centered at $v_i$ and the visibility kernel of $P_i$ intersects only the edges of the kernel incident to $v_i$; call this intersection $S_i$.  A specific position $u_i$ is then chosen within each region $S_i$.

A close examination of their paper reveals that the strategy for vertices of degree $5$ generalizes to vertices $u$ of arbitrary degree (greater than $2$).  In particular, the definition of $u_i$ depends solely on $\eps$ and the positions of the edges in $P_i$ incident to $v_i$.

The radius $\eps$ is computed as follows.  For each index $i$, we need a positive distance $\eps_i > 0$ smaller than the minimum distance from $v$ to any edge of the kernel of $P_i$ that is not incident to $v$, at any time during the morphing step $\Gamma_i \leadsto \Gamma_{i+1}$.  It suffices to compute the minimum distance from $v$ to the lines supporting edges of $P_i$ not incident to $v$.  The squared distance to each of these lines at any time $t$ can be expressed as the ratio $f(t)/g(t)$ of two quadratic polynomials $f$ and $g$.  Alamdari et al.~argue that a lower bound $0 < \delta \le \min_t \sqrt{f(t)/g(t)}$ can be computed in constant time in an appropriate real RAM model~\cite{aabcd-hmpgd-17}.  Then $\eps_i$ is the minimum of these lower bounds $\delta$.  Altogether computing~$\eps_i$ takes $O(\deg(u))$ time.  Finally, $\eps = \min_{1 \le i \le k} \eps_i$.

Once the radius $\eps$ is known, computing each sector $S_i$ in $O(\deg(u))$ time is straightforward.

The point $u_0$ can be chosen arbitrarily within $S_0$.  For each index $i$ in increasing order, Alamdari et al. describe how to choose a point $u_{i+1}\in S_{i+1}$ in $O(1)$ time, such that the vector $u_i\arcto u_{i+1}$ is parallel to the vector $v_i\arcto v_{i+1}$.  This part of the algorithm makes no reference to the rest of the triangulation; it works entirely within the sectors $S_i$ and $S_{i+1}$.  Moreover, no part of this algorithm relies on $u$ being \emph{directly} collapsed to $v$, only that vertices $u$ and $v$ have the same image in the triangulations $\Gam'_i$ and $\Gam'_{i+1}$.

We apply this perturbation technique to the toroidal pseudomorphs computed in Section~\ref{ssec:triangulationpseudomorph} as follows.  Recall that our pseudomorph consists of a direct collapse, a recursively computed pseudomorph, and a reversed spring collapse.  First we (recursively) perturb the recursive pseudomorph into a proper morph $\Gam'_0\leadsto\cdots\leadsto\Gam'_k$ consisting of $k$ parallel linear morphs.  We then compute the sectors~$S_i$ and the radius $\eps$ exactly as described above.  To perturb the initial direct collapse from $u$ to $v$, we move $u$ to an arbitrary point in the intersection of $S_0$ and the edge $e$ along which $u$ is collapsed.  We compute the intermediate positions $u_i$ for $u$ exactly as described above, working entirely within the local coordinates of the sectors $S_i$.  Finally, to perturb the reversed spring collapse, we first move $u$ from $u_k$ to a new point $u'_k \in S_k$ so that the image of the collapsing edge $e$ becomes parallel to the direction of the original spring collapse, after which we simply interpolate to the final triangulation, as described in Section \ref{SS:interp}.  
Because the vertices move along parallel geodesics,  Lemma~\ref{lem:paralleltriangles} implies that this final interpolation is a parallel linear morph.  Altogether, we obtain a morph consisting of $k+3$ parallel linear morphs.  We emphasize that the additional step moving $u_k$ to $u'_k$ is the only significant difference from the algorithm presented by Alamdari et al.

Unrolling the recursion, we can perturb our pseudomorph between two $n$-vertex toroidal triangulations into a proper morph consisting of $O(n)$ parallel linear morphing steps in $O(n^2)$ time.  The overall time to compute this morph is still dominated by the time needed to compute $O(n)$ equilibrium triangulations for the spring collapses. 

\begin{theorem}
\label{thm:depseudo}
Given any two isotopic geodesic toroidal triangulations $\Gam_0$ and $\Gam_1$ with $n$ vertices, we can compute a morph from $\Gam_0$ to $\Gam_1$ consisting of $O(n)$ parallel linear morphs in $O(n^{1+\omega/2})$ time.
\end{theorem}

\section{Not Just Triangulations}
\label{sec:allgraphs}

Finally, it remains to describe how to morph between embeddings that are not triangulations.  Following existing work in the planar setting, we extend the given embeddings $\Gam_0$ and $\Gam_1$ to triangulations, and then invoke our earlier triangulation-morphing algorithm.  The main difficulty is that it may not be possible to triangulate both $\Gam_0$ and $\Gam_1$ using the same diagonals, because corresponding faces, while combinatorially identical, have different shapes.

Two different techniques have been proposed to overcome this hurdle in the planar setting.  The first method subdivides each pair of corresponding faces into a  compatible triangulation, introducing additional vertices if necessary~\cite{t-dpg-83,gs-gipm-01,sg-cmcpt-01,sg-msfuo-01,aacbf-mpgdw-13}; however, this technique increases the complexity of the graph to $O(n^2)$~\cite{ass-ctsp-93}.  The second technique uses additional morphing steps to convexify faces to that they can be compatibly triangulated without additional vertices~\cite{addfp-mpgdo-14, aabcd-hmpgd-17}.  While the subdivision technique generalizes to toroidal embeddings (at least when all faces are disks), it is unclear how to generalize existing morphing techniques.

We introduce a third technique, which avoids both subdivision and additional morphing steps by exploiting Theorem~\ref{Th:tutte-torus}.  We emphasize that our method can also be applied to 3-connected straight-line plane graphs, giving a new and arguably simpler approach for the planar case as well.

\begin{theorem}
\label{Th:main}
Given any two isotopic essentially 3-connected geodesic toroidal embeddings $\Gam_0$ and $\Gam_1$ with $n$ vertices, we can compute a morph from $\Gam_0$ to $\Gam_1$ consisting of $O(n)$ parallel linear morphs in $O(n^{1+\omega/2})$ time. 
\end{theorem}

\begin{proof}
Let $\Gam_*$ be an equilibrium embedding isotopic to $\Gam_0$ and $\Gam_1$ as given by Theorem~\ref{Th:tutte-torus}.
It suffices to describe how to morph from $\Gam_0$ to $\Gam_*$; to morph from $\Gam_0$ to $\Gam_1$ one can simply first morph from $\Gam_0$ to $\Gam_*$ and then from $\Gam_*$ to $\Gam_1$.

\emph{Arbitrarily} triangulate the faces of~$\Gam_0$; this can be done in $O(n)$ time using Chazelle's algorithm \cite{c-tsplt-91}, or in $O(n\log n)$ time in practice.  Because each face of $\Gam_*$ is convex, we can triangulate $\Gam_*$ in the exact same manner. The result is two isotopic geodesic toroidal triangulations~$T_0$ and $T_*$. Given a morph between $T_0$ and $T_*$ as promised by Theorem~\ref{thm:depseudo}, we obtain a morph between $\Gam_0$ and $\Gam_*$ by simply ignoring the edges added when triangulating.  In particular, the morph is specified by a sequence of geodesic triangulations $T_0, T_1, \ldots, T_k=T_*$, and dropping the additional edges from each triangulation~$T_i$ results in a geodesic embedding $\Gam_i$ isotopic to $\Gam_0$.

The number of parallel morphing steps remains $O(n)$, and the running time is dominated by the computation of the morph between $T_0$ and $T_*$, which is $O(n^{1+\omega/2})$ by Theorem~\ref{thm:depseudo}.
\end{proof}

Finally, Theorem \ref{Th:main} immediately yields the first proof of the following corollary:

\begin{corollary}
Two essentially 3-connected geodesic embeddings on the flat torus are isotopic if and only if they are isotopic through geodesic embeddings.
\end{corollary}

\section{Conclusions and Open Questions}
\label{sec:conclusions}

In this paper, we have given the first algorithm to construct a morph between two isotopic geodesic graphs on the flat torus.  Key tools in our algorithm include a geometric analysis of 6-regular triangulations on the torus, as well as repeated use of a generalization of Tutte's spring embedding theorem by Y.~Colin de Verdière~\cite{c-crgtd-91} (Theorem \ref{Th:tutte-torus}).  Several of our applications of spring embeddings also apply to planar morphs and give a new approach to compute linear complexity morphs.

Because it relies heavily on Theorem \ref{Th:tutte-torus}, our algorithm requires that the input embeddings are essentially 3-connected.  If a given toroidal embedding $\Gamma$ is \emph{not} essentially 3-connected, an isotopic equilibrium drawing $\Gamma_*$ still exists, but it may not be an embedding; nontrivial  subgraphs can collapse to geodesics or even points in $\Gamma_*$.  Morphing less connected toroidal embeddings remains an open problem.

It is natural to ask whether any of our results can be extended to higher-genus surfaces.  Y.~Colin de Verdière actually generalized Tutte’s theorem to graphs on arbitrary Riemannian 2-manifolds without positive curvature~\cite{c-crgtd-91}.  However, the resulting equilibrium embedding is the solution a certain convex optimization problem that, in general, cannot be formulated as solving a linear system.  At a more basic level, our analysis of cats and dogs in Section~\ref{sec:nobadtriangulations} relies on the average vertex degree being exactly $6$, a property that holds \emph{only} for graphs on the torus or the Klein bottle.  Even the existence of “parallel” morphs requires a locally Euclidean metric.

Our results share two closely related limitations with existing planar morphing algorithms.  First, planar morphs involving either Cairns-style edge collapses or spring embeddings require high numerical precision to represent exactly.  Second, while edge collapses yield planar morphs with low combinatorial complexity, the resulting morphs are not good for practical visualization applications \cite{lp-mpgdb-11}.  Because our algorithm uses both edge collapses and spring embeddings, it suffers from the same numerical precision issues and (we expect) the same practical limitations.  Floater, Gotsman, and Surazhsky’s barycentric interpolation technique \cite{fg-mti-99,sg-msfuo-01,gs-gipm-01,sg-cmcpt-01} yields better results in practice for planar morphs, but as we discussed in Section~\ref{ssec:prior}, their technique does not immediately generalize to the torus.

Finally, there are several variants and special cases of planar morphing for which generalization to the flat torus would be interesting, including morphing with bent edges \cite{lp-mpgdb-11}, morphing orthogonal embeddings \cite{bls-mpgwp-05,blps-mopgd-13,gv-ompod-18,gsv-ompod-19}, and morphing weighted Schnyder embeddings \cite{bhl-msdpt-19}.

\paragraph*{Acknowledgements.} We thank Timothy Chan for inspiring this work by noting to the third author that extensions of Alamdari et al.~\cite{aabcd-hmpgd-17} to the torus were unknown.

\bibliographystyle{newuser-doi} 
\bibliography{morph}

\appendix
\section{Relaxed Coordinate Representations}
\label{A:coordinates}

The coordinate representation described in Section \ref{SS:canon}, while intuitive, is much more constrained than necessary.  Here we describe a more relaxed representation that allows morphs to be described entirely in terms of changing vertex coordinates, while still encoding the non-trivial topology of the embedding.  Similar representations have been traditionally used to model periodic (or “dynamic”) graphs \cite{kmw-ocure-67, c-spg-78, ks-dcdgp-88, cm-spadc-93, o-spdg-84, rk-riati-88, w-pdmia-67, is-tcigw-87, is-scpzs-90, i-tdgtv-87}, and more recently to model periodic bar-and-joint frameworks \cite{bs-pff-10,bs-lsppf-15}.

We also observe that this representation leads naturally to a linear-time algorithm to test whether two given toroidal embeddings are isotopic.  Our algorithm is arguably simpler than the more general linear-time isotopy algorithm of É.~Colin de Verdière and de Mesmay \cite{cm-tgis-14}.  We emphasize that both of these isotopy algorithms are non-constructive; they do not construct an isotopy if one exists.  Rather, both algorithms test topological conditions that characterize isotopy \cite{l-ctp1c-74a,l-ctp1c-74b,l-cdp1c-84}.

\subsection{Translation Vectors and Equivalence}

To represent a geodesic embedding $\Gamma$ of a graph $G$ on the flat torus $\Torus$, we associate a \EMPH{coordinate vector} $p(v) \in \Real^2$ with every vertex $v$ of $G$ and a \EMPH{translation vector} $\tau(d) \in \Z^2$ with every dart $d$ of $G$.   We do \emph{not} require vertex coordinates to lie in the unit square; instead, each coordinate vector $p(v)$ records the coordinates of an arbitrary lift $\widetilde{v}$ of $v$ to the universal cover $\widetilde{\Gamma}$.  The translation vector of each dart encode which lifts of its endpoints are connected in $\widetilde{\Gamma}$.  Specifically, for each dart $d$ in $G$, the universal cover $\widetilde{\Gamma}$ contains an edge between $p(\Tail(d))$ and $p(\Head(d)) + \tau(d)$, and therefore also contains and edge  between $p(\Tail(d)) + (i,j)$ and $p(\Head(d)) + \tau(d) + (i,j)$ for all integers $i$ and $j$.  Translation vectors are antisymmetric: $\tau(d) = -\tau(\Rev(d))$.

Coordinate representations are not unique; in fact, each toroidal embedding has an infinite family of equivalent representations.  Two coordinate representations $(p, \tau)$ and $(p', \tau')$ with the same underlying graph are \emph{equivalent}, meaning they represent the same geodesic embedding (up to translation), if and only if
\[
	\Delta(d)
	~:=~ p(\Head(d)) + \tau(d) - p(\Tail(d))
	~=~ p'(\Head(d)) + \tau'(d) - p'(\Tail(d))
\]
for every dart $d$.  The vector $\Delta(d)$, which we call the \EMPH{displacement vector} of $d$, is the difference between the head and tail of any lift of $d$ to $\widetilde{\Gamma}$.

Let $(p, \tau)$ be any coordinate representation of $\Gamma$.  Given \emph{arbitrary} integer vector $\pi(v)\in \Z^2$ for each vertex of $G$, we can define a new coordinate representation $(p^\pi, \tau^\pi)$ as follows:
\begin{align*}
	p^\pi(v) &= p(v) + \pi(v) & \text{for every vertex $v$} \\
	\tau^\pi(d) &= \tau(d) + \pi(\Tail(d)) - \pi(\Head(d)) & \text{for every dart $d$}
\end{align*}
Easy calculation implies that the representations $(p, \tau)$ and $(p^\pi,\tau^\pi)$ are equivalent.  This transformation is a multidimensional generalization of the \emph{reweighting} or \emph{repricing} strategy proposed by Tomizawa \cite{t-stust-71} and Edmonds and Karp \cite{ek-tiaen-72} for minimum-cost flows, and later used by Johnson to compute shortest paths \cite{j-easps-77}.

Every geodesic toroidal embedding has a unique \EMPH{canonical} coordinate representation $(p, \tau)$, where $p(v) \in [0,1)^2$ for every vertex $v$.  In this canonical coordinate representation, each translation vector $\tau(d)$ encodes how  dart $d$ crosses the boundaries of the fundamental square; in other words, canonical translation vectors are crossing vectors, exactly as described in Section \ref{SS:canon}. 

\subsection{Normalization and Isotopy Testing}

Let $\Gamma_0$ and $\Gamma_1$ be two isotopic geodesic toroidal embeddings of the same graph $G$, given by coordinate representations $(p_0, \tau_0)$ and $(p_1, \tau_1)$ respectively.  To simplify the presentation of our morphing algorithm, we implicitly assume that the translation vectors in both representations are identical: $\tau_0(d) = \tau_1(d)$ for every dart $d$.  This assumption allows us to describe, reason about, and ultimately compute a morph from~$\Gamma_0$ to $\Gamma_1$ entirely in terms of changes to the vertex coordinates; all translation vectors remain fixed throughout the morph.

If necessary, we can enforce this assumption in $O(n)$ time using the following \emph{normalization} algorithm.  Let $(p_0, \tau_0)$ and $(p_1, \tau_1)$ be the given coordinate representations of $\Gamma_0$ and $\Gamma_1$, respectively.  First, construct an arbitrary spanning tree $T$ of the underlying graph $G$, directed away from an arbitrary root vertex $r$.  For every vertex $v$, let $P(v)$ denote the unique directed path in $T$ from $r$ to $v$.  For each vertex $v$, let
\[
	\pi(v) = \sum_{d\in P(v)} (\tau_1(d) - \tau_0(d)).
\]
We can compute the vectors $\pi(v)$ for all vertices in $O(n)$ time by preorder traversal of $T$.  Finally, we replace the target representation $(p_1, \tau_1)$ with the equivalent representation $(p_1^\pi, \tau_1^\pi)$.

\begin{lemma}
For all darts $d$ in $T$, we have $\tau_1^\pi(d) = \tau_0(d)$.
\end{lemma}

\begin{proof}
Let $d$ be any dart in $T$ directed from some vertex $u$ to one of its children $v$ in $T$.  Straightforward calculation implies
\begin{align*}
	\tau_1^\pi(d) &= \tau_1(d) + \pi(u) - \pi(v)
	\\
	& = \tau_1(d) + \sum_{d'\in P(u)} (\tau_1(d') - \tau_0(d'))
		- \sum_{d'\in P(v)} (\tau_1(d') - \tau_0(d'))
	\\
	& = \tau_1(d) - (\tau_1(d) - \tau_0(d)) \\
	& = \tau_0(d).
\end{align*}
A similar calculation (or antisymmetry) implies that $\tau_1^\pi(d) = \tau_0(d)$ for every dart $d$ directed from a vertex to its parent in $T$.
\end{proof}

\begin{theorem}
\label{Th:isotopy}
$\Gamma_0$ and $\Gamma_1$ are isotopic if and only if $\tau_1^\pi(d) = \tau_0(d)$ for every dart $d$.
\end{theorem}

\begin{proof}
We exploit a theorem of Ladegaillerie \cite{l-ctp1c-74a,l-ctp1c-74b,l-cdp1c-84}, which states that two embeddings are isotopic if and only if every cycle in one embedding is homotopic to the corresponding cycle in the other embedding.  

For any dart $d$ in the underlying graph $G$, let $\Delta_0(d)$ and $\Delta_1(d)$ denote the displacement vectors of $d$ in $\Gamma_0$ and $\Gamma_1$, respectively.  (We emphasize that displacement vectors are independent of the coordinate representation.)   For any directed cycle $C$ in $G$, let $\Delta_0(C)$ and $\Delta_1(C)$ denote the sum of the displacement vectors of its darts:
\begin{align*}
	\Delta_0(C) &~:=~ \sum_{d\in C} \Delta_0(d)
	&
	\Delta_1(C) &~:=~ \sum_{d\in C} \Delta_1(d)
\end{align*}
The vector $\Delta_0(C)$ is the \emph{integer homology class} of $C$ in $\Gamma_0$.  Two cycles on the torus are homotopic if and only if they have the same integer homology class; in particular, the image of $C$ in $\Gamma_0$ is contractible if and only if $\Delta_0(C) = (0,0)$.  Ladegaillerie's theorem implies that $\Gamma_0$ and $\Gamma_1$ are isotopic if and only if $\Delta_0(C) = \Delta_1(C)$ for every cycle $C$.

The spanning tree $T$ defines a set of \emph{fundamental cycles} that span the cycle space of $G$.  Specifically, for each dart $d$ that is not in $T$, the fundamental directed cycle $C_T(d)$ consists of $d$ and the unique directed path in $T$ from $\Head(d)$ to $\Tail(d)$.  Every directed cycle in $G$ (indeed every \emph{circulation} in $G$) can be expressed as a linear combination of fundamental cycles.  It follows by linearity that $\Gamma_0$ and $\Gamma_1$ are isotopic if and only if every \emph{fundamental} cycle has the same integer homology class in both embeddings; that is, $\Delta_0(C_T(d)) = \Delta_1(C_T(d))$ for every dart $d\in G\setminus T$.

Straightforward calculation implies that the homology class of any cycle is also equal to the sum of the \emph{translation} vectors of its darts with respect to \emph{any} coordinate representation:
\begin{align*}
	\Delta_0(C) &~=~ \sum_{d\in C} \tau_0(d)
	&
	\Delta_1(C) &~=~ \sum_{d\in C} \tau_1(d) ~=~ \sum_{d\in C} \tau_1^\pi(d).
\end{align*}
In particular, for any non-tree dart $d\not\in T$, we immediately have
\begin{align*}
	\Delta_0(C_T(d)) - \Delta_1(C_T(d))
		~=~ \sum_{d'\in C_T(d)} \left(\tau_0(d') - \tau_1^\pi(d')\right)
		~=~ \tau_0(d) - \tau_1^\pi(d)
\end{align*}
Thus, $\tau_1^\pi(d) = \tau_0(d)$ for every dart $d$ if and only if $\Delta_0(C) = \Delta_1(C)$ for every fundamental cycle $C$, which completes the proof of the theorem.
\end{proof}

Theorem \ref{Th:isotopy} and our normalization algorithm immediately imply an $O(n)$-time algorithm to test whether two given coordinate representations $(p_0, \tau_0)$ and $(p_1, \tau_1)$ represent isotopic toroidal embeddings of the same graph~$G$.  Our algorithm is arguably simpler than the isotopy algorithm of É.~Colin de Verdière and de Mesmay \cite{cm-tgis-14}, which is also based on Ladegaillerie's theorem \cite{l-ctp1c-74a,l-ctp1c-74b,l-cdp1c-84}.  On the other hand, our isotopy algorithm is specific to geodesic embeddings on the flat torus; whereas, É.~Colin de Verdière and de Mesmay's algorithm works for arbitrary combinatorial embeddings of graphs on arbitrary 2-manifolds.

\end{document}